\documentclass[acmsmall,nonacm]{acmart}
\usepackage{booktabs}
\usepackage{pgfplots}
\pgfplotsset{compat=1.18}
\usepgfplotslibrary{fillbetween,groupplots}
\usetikzlibrary{arrows.meta,shapes.geometric}

\usepackage{amsmath,amssymb,amsfonts,mathtools}
\usepackage{amsthm}
\usepackage{nicefrac}
\usepackage{xcolor}
\usepackage{xspace}
\usepackage{enumitem}
\usepackage{algorithm}
\usepackage{algorithmic}
\usepackage{stmaryrd}
\usepackage{url}

\theoremstyle{plain}
\newtheorem{theorem}{Theorem}[section]
\newtheorem{lemma}[theorem]{Lemma}
\newtheorem{proposition}[theorem]{Proposition}
\newtheorem{corollary}[theorem]{Corollary}

\theoremstyle{definition}
\newtheorem{definition}[theorem]{Definition}
\newtheorem{example}[theorem]{Example}

\theoremstyle{remark}
\newtheorem{remark}[theorem]{Remark}

\newcommand\ma[1]{\ensuremath{\mathcal{#1}}}

\newcommand*\eg{\text{e.g.,}\xspace}

\newcommand\tx[1]{\texttt{#1}}

\newcommand*\query{\ensuremath{q}}

\newcommand*\cqs{\ensuremath{\mathit{CQs}}\xspace}

\newcommand*\cq{\ensuremath{\mathit{CQ}}\xspace}

\newcommand*\size[1]{\ensuremath{|#1|}}

\newcommand*\attribute[1]{\ensuremath{\operatorname{att}(#1)}\xspace}

\newcommand*\closedComponent[1]{\ensuremath{\operatorname{CC}(#1)}\xspace}

\newcommand*\JS{\ensuremath{\operatorname{JS}}\xspace}
\newcommand*\KL{\ensuremath{\operatorname{KL}}\xspace}

\newcommand*\ov{\ensuremath{\mathit{ov}}\xspace}

\newcommand*\fd{\ensuremath{\mathit{FD}}\xspace}

\newcommand*\coverage{\ensuremath{\operatorname{Cover}}\xspace}

\newcommand{\Obs}{\ensuremath{\operatorname{Obs}_{\Sigma,\Omega}}}
\newcommand{\ObsFam}{\ensuremath{\Gamma_{\Sigma,\Omega}}}

\newcommand{\AttrUniverse}{\ma V}

\newcommand{\LegalWorlds}{\ma W}

\newcommand{\LawSet}{\Sigma}

\newcommand{\OverlapSet}{\Omega}

\newcommand{\AttrClosure}[2][\LawSet]{{#2}^+_{#1}}

\newcommand{\AugSchema}[1]{\widetilde{#1}}

\newcommand{\CCGraph}{\ma G_{\LawSet,\OverlapSet}}

\newcommand{\AtomObligs}{\ma B_Q}

\newcommand{\CandSet}{\ma A}

\newcommand{\GreedySol}{\ma A_{\mathrm{greedy}}}

\newcommand{\OptSol}{\ma A_{\mathrm{opt}}}

\usepackage{orcidlink}   
\noexpandarg             
\title{Identifiability of Relational Queries in Multi-View Pretraining}

\author[Ratan Bahadur Thapa]{Ratan Bahadur Thapa\,\texorpdfstring{\orcidlink{0009-0000-2368-5928}}{}}
\orcid{0009-0000-2368-5928}
\affiliation{%
  \department{Analytic Computing, KI}
  \institution{University of Stuttgart}
  \city{Stuttgart}
  \country{Germany}
}
\email{ratan.thapa@ki.uni-stuttgart.de}

\author[Daniel Hern\'{a}ndez]{Daniel Hern\'{a}ndez\,\texorpdfstring{\orcidlink{0000-0002-7896-0875}}{}}
\affiliation{%
  \department{Analytic Computing, KI}
  \institution{University of Stuttgart}
  \city{Stuttgart}
  \country{Germany}
}
\email{daniel.hernandez@ki.uni-stuttgart.de}
\orcid{0000-0002-7896-0875}
\fullexpandarg           

\begin{abstract}
When data sources are integrated through a shared interface, a downstream query may or may not be \emph{determined} by what the interface exposes: two globally consistent worlds can agree on every shared attribute yet disagree on the query answer. This ambiguity is structural---a property of the interface design, not the data volume---and cannot be resolved by collecting more records or training a larger model.

We formalize \emph{query identifiability} for data integration under interface laws---functional dependencies that hold uniformly across all legal worlds rather than within a single instance---and prove three results. (i)~A polynomial-time certificate (\tx{CheckCert}) decides identifiability via attribute closure, and is exact on instances that actually expose any residual ambiguity (which we formally call \emph{closure-separable}). (ii)~Non-identifiable queries face an irreducible $\nicefrac{1}{2}$ minimax error floor for any estimator using only interface evidence, bounding multi-view pretraining systems from below. (iii)~A minimum-augmentation algorithm (\tx{Greedy-MinAug}) finds the smallest set of interface additions to certify a query, reducing the problem to \textsc{Set Cover} (logarithmic approximation).

Experiments on synthetic benchmarks, real integration datasets spanning three domains (scholarly, product, and restaurant), and schemas up to $10^3$ attributes confirm \tx{CheckCert} is exact, both algorithms run in single-digit milliseconds, and ML classifiers exhibit the predicted error floor and abrupt capability gains that augmentation produces.
\end{abstract}



\begin{document}
\maketitle

\section{Introduction}
\label{sec:introduction}
Enterprise analytics increasingly draws from data lakes and federated ecosystems where analysts query across hundreds to thousands of sources with incompatible schemas~\cite{nargesian2019datalake,stonebraker2018integration}; at open-data scale, catalogues index tens of millions of datasets from independent producers~\cite{brickley2019datasetsearch}. In each case an integration layer specifies which attributes are shared across sources and which deterministic correspondences between them hold, but leaves the rest unconstrained. A growing class of representation-learning methods, \emph{multi-view} and contrastive pretraining, learns directly from these shared attributes, aligning records across sources by the attributes they share~\cite{li2018survey,zhang2025augmentation}. This raises a foundational question \emph{prior} to any particular model or training procedure: given only what the integration interface exposes, which queries are \emph{determined}---guaranteed the same answer across every dataset consistent with the interface---and which are irreducibly ambiguous, no matter how much data is gathered or how large the model? We call this property \emph{query identifiability}, and it is the subject of this paper.

Where does this ambiguity arise? Consider a billing system and a support system integrated through a shared interface. A data engineer needs to answer a query: \emph{does this customer have both an overdue invoice and a high-severity support ticket?} The billing system identifies customers by email address (\tx{email}); the support system uses an internal customer code (\tx{cid}). Their shared interface aligns records on \tx{email}---the one attribute they have in common---but says nothing about which \tx{cid} corresponds to which email. Without that correspondence the interface offers no unique way to match invoices with tickets: there are multiple ways to complete the data consistently with everything the interface exposes, and they can assign \tx{cid} values differently, linking tickets to invoices differently and returning different query answers. This ambiguity is \emph{structural}: it cannot be resolved by collecting more data, because the interface simply does not carry the \tx{email}-to-\tx{cid} mapping. Add a resolver that fixes that mapping as a deterministic law, however, and the answer changes entirely---the resolver propagates the email evidence to \tx{cid}, the join key becomes available, and the query is determined by the interface alone.

To reason about this systematically, we model every way of completing the observed data as a \emph{legal world}---a globally consistent dataset that respects the interface's deterministic components. A query is \emph{identifiable} when all legal worlds return the same answer~\cite{abiteboul1995foundations,fagin1983semantics}. The deterministic components---resolvers, canonicalizers, identifier crosswalks---are modeled as \emph{interface laws}: constraints that fix one attribute from another uniformly across all legal worlds~\cite{armstrong1974dependency,beeri1977complete}. These laws propagate overlap evidence to further attributes: in our billing example, the $\tx{email}\to\tx{cid}$ resolver law reaches \tx{cid} from the email overlap, making the join key available and the query identifiable.

This structure yields three results. First, whether the laws propagate far enough to cover all attributes a query needs is decidable in polynomial time (\tx{CheckCert}, Theorem~\ref{thm:frontier-suff}); on instances that actually expose the ambiguity, a failed check is itself a proof of non-identifiability (Corollary~\ref{cor:closure-complete}). Second, for every non-identifiable query there is always a concrete pair of legal worlds on which any method using only interface evidence must err with probability at least $\nicefrac{1}{2}$, regardless of how much data it sees (Theorem~\ref{thm:nonident-minimax})---the information is structurally absent. Third, adding an interface component---a new resolver, a crosswalk, a shared identifier---can flip a non-identifiable query to identifiable, producing an abrupt capability gain (Theorem~\ref{thm:connectivity-emergence}); \tx{Greedy-MinAug} finds the smallest such addition, reducing the problem to \textsc{Set Cover} with a logarithmic approximation (Theorem~\ref{thm:minaug-hard-greedy}). Because multi-view pretraining minimises a loss anchored to the shared attributes, driving that loss below a threshold forces the unique answer on certified queries---structurally determined, independently of sample size (Theorem~\ref{thm:robust-threshold}).

\paragraph{Contributions.}
We turn these results into two practical schema-design algorithms. \tx{CheckCert} checks in polynomial time whether interface propagation covers all attributes a query needs---certifying identifiability on instances that expose the ambiguity (Corollary~\ref{cor:closure-complete})---and pinpoints unreachable attributes when the check fails. \tx{Greedy-MinAug} finds the smallest set of new interface components that certifies the query, reducing the problem to \textsc{Set Cover} with a logarithmic approximation (Theorem~\ref{thm:minaug-hard-greedy}). We evaluate both algorithms on a 5-attribute synthetic benchmark with exhaustive ground truth, on real integration datasets spanning three domains (scholarly: BibInteg and CrossKG-DBLP; product: Amazon-Google; restaurant: Fodors-Zagat) with the WDC schema as a design case study, and at database scale (up to a thousand attributes and functional dependencies). The certificate is exact and both algorithms run in single-digit milliseconds. ML classifiers trained on the same schemas confirm the $\nicefrac{1}{2}$ error floor and the abrupt capability gains that augmentation produces. The paper is organized as follows: Section~\ref{sec:related-work} discusses related work; Section~\ref{sec:model}--\ref{sec:algorithms} develop the theory and algorithms; Section~\ref{sec:experiments} presents the evaluation. Key theoretical results are machine-verified in Lean~4: the closure certificate (Theorem~\ref{thm:frontier-suff}), minimax lower bound (Theorem~\ref{thm:nonident-minimax}), and the reduction of identifiability to query determinacy are fully machine-checked; the capability-jump theorem and MinAug hardness remain open in the formalization.
\section{Related Work}
\label{sec:related-work}

Possible-worlds semantics and certain answers provide the closest semantic template \cite{abiteboul1995foundations,fagin1983semantics,libkin2011incomplete}. Evidence restricts a set of admissible worlds, and a query is certain when its answer is invariant across those worlds. We adapt the idea to heterogeneous data integration by replacing a single incomplete database with observed views, designated overlaps, and deterministic interface laws that hold across all legal worlds.

Data integration studies query answering over heterogeneous sources with different schemas and identifiers \cite{halevy2006data,doan2012principles,lenzerini2002data}. View-based query answering and query determinacy ask when views suffice to answer or rewrite a query \cite{levy1995answering,halevy2001answering,nash2010views,abiteboul1998complexity,pasaila2011conjunctive}. Nash, Segoufin, and Vianu~\cite{nash2010views} give a semantic characterisation of query determinacy, prove it undecidable for first-order views, and show decidability only in restricted \cq\ cases (monadic/Boolean and path queries); the general \cq\ case was later proved undecidable over both unrestricted~\cite{gogacz2015redspider} and finite~\cite{gogacz2016rainworm} instances. Our identifiability is the sub-case where observation views are closure-augmented overlap projections---projection \cqs with no joins---which falls within the decidable fragment; two further restrictions yield polynomial time (Theorem~\ref{thm:frontier-suff}): the views are FD-closures of designated overlaps, and interface laws hold \emph{across} legal worlds (Definition~\ref{def:legality}) rather than within a single instance, reducing the containment check to standard attribute closure with no quantifier alternation. Identifiability therefore does not subsume \cq\ views with joins; conversely, determinacy by \cq\ join views is undecidable~\cite{gogacz2015redspider,gogacz2016rainworm} and admits no closure certificate.

\begin{table}[t]
\caption{Identifiability vs.\ related query-answering frameworks (Boolean CQs). $^\dagger$Complete on closure-separable instances (Corollary~\ref{cor:closure-complete}).}
\label{tab:related}
\centering
\setlength{\tabcolsep}{4pt}
\begin{tabular}{@{}p{2.4cm}p{2.5cm}p{2.9cm}@{}}
\toprule
Framework & CQ complexity & Certificate \\
\midrule
Query determinacy~\cite{nash2010views} & Undecidable (general) & Sufficient only \\
Certain answers~\cite{abiteboul1995foundations,libkin2011incomplete} & \textsc{coNP}-complete & Complete \\
Data exchange~\cite{fagin2005data} & \textsc{PTIME} (chase) & Sound (target-side) \\
\tx{CheckCert} (Our work) & \textsc{PTIME} (closure) & Complete$^\dagger$ \\
\bottomrule
\end{tabular}
\end{table}

The tractability gap in Table~\ref{tab:related} has a precise source. Query determinacy asks whether an arbitrary set of views rewrites $\query$ under arbitrary world semantics---no constraint links views to each other, so the problem reduces to checking containment for all possible view extensions, which is undecidable in general~\cite{nash2010views} and undecidable even for \cqs~\cite{gogacz2015redspider,gogacz2016rainworm}. Certain answers are complete but \textsc{coNP}-hard because every possible open-world completion must be considered. Interface laws break both barriers: once evidence is anchored to designated overlaps and laws are restricted to FD-syntax holding \emph{across worlds}, propagation reduces to standard attribute closure---a linear fixed-point computation with no quantifier alternation. Identifiability is, at its core, a closure problem; this paper recognizes that structure, formalizes it under multi-world semantics, and derives its consequences for pretraining and schema design.

Data exchange studies legal target instances and certain answers under schema mappings \cite{fagin2005data}. The chase procedure produces canonical instances that satisfy a set of tgds and egds; interface laws are a restricted form of egds that hold universally. Minimum interface augmentation (Definition~\ref{def:minaug}) is analogous to strengthening a schema mapping to eliminate ambiguity in the target: adding interface actions corresponds to adding egd constraints that make the target unique for more queries. Entity resolution and record linkage provide practical mechanisms for constructing cross-source identity evidence \cite{christen2012data,dong2009data}. We abstract from a particular matching algorithm and ask which queries become determined once an interface exposes particular overlaps and deterministic components.

Functional dependencies and attribute closure are classical tools for reasoning about determinacy of attributes \cite{armstrong1974dependency,beeri1977complete,ginsburg1983characterizations}. We use the same implication machinery, but our dependencies represent deterministic interface laws shared across legal worlds rather than only constraints internal to one relation instance. Minimum interface augmentation is closest to view-set minimization and view-selection problems \cite{li2001minimizing}.

Conjunctive queries are the canonical language for joins and projections in database theory \cite{chandra1977optimal,abiteboul1995foundations}. We use their attribute footprint to connect query evaluation to closure-augmented overlap evidence. Jensen--Shannon divergence \cite{lin2002divergence}, information inequalities \cite{cover2006elements}, and Fano-style lower bounds \cite{fano1961transmission,tsybakov2009introduction,le2012asymptotic} provide the tools for our robustness and capacity results.

Multi-view learning uses agreement between different observations as surrogate supervision \cite{blum1998combining,andrew2013deep,li2018survey}. Representation-learning pipelines implement the same principle through contrastive or alignment losses on paired views. The key distinction from this paper is \emph{semantic vs.\ statistical}: multi-view learning bounds concern generalisation---whether a model trained on $N$ samples can predict well---while our identifiability is a structural question independent of sample size. A query can be statistically learnable (with enough data, a model converges to the right answer) yet not identifiable in our sense (the information is simply not present in the interface evidence). Theorem~\ref{thm:nonident-minimax} makes this precise: for non-identifiable queries, there exists a witness pair on which any interface-evidence-only estimator achieves error $\ge\tfrac{1}{2}$ regardless of $N$, whereas generalisation bounds decrease with $N$. Identifiability is therefore a precondition that learning bounds implicitly assume; our contribution is to make this precondition explicit and algorithmically checkable.

\section{Data Integration Interface}
\label{sec:model}

Let us introduce a running example to help formalize three key notions: \emph{interface}, \emph{worlds}, and \emph{legality}.

\begin{example}[Running example]
  A company integrates two sources (Table~\ref{tab:re-schema}). \textsc{Billing} records, for each customer \tx{email}, an \tx{invoice} and whether it is \tx{overdue}; \textsc{Support} records, for each internal identifier \tx{cid}, a \tx{ticket} and its \tx{severity}. The two share no attribute directly---Billing is keyed by \tx{email}, Support by \tx{cid}---so the query $\query$ (\emph{does a customer with an overdue invoice also have a high-severity ticket?}) can be answered only by linking \tx{email} to \tx{cid}. An \textsc{HR} \emph{resolver} supplies that link: the rule $\tx{email}\to\tx{cid}$. With the resolver the link is pinned; without it the data is ambiguous---Table~\ref{tab:two-worlds} shows two datasets $w$ and $w'$ that report the same \textsc{Billing} and \textsc{Support} records yet pair customers with tickets oppositely, so \tx{a@x}'s overdue invoice meets a high-severity ticket in $w$ ($\query=\mathrm{true}$) but a low-severity one in $w'$ ($\query=\mathrm{false}$). Same observations, two answers.
\end{example}

\begin{table}[t]
\caption{Running example. \textsc{Billing} and \textsc{Support} are the integrated \emph{sources}; \textsc{HR} is the \emph{resolver}, which supplies the interface law $\tx{email}\to\tx{cid}$ that links them. The resolver is optional---its presence is what makes the query identifiable.}
\label{tab:re-schema}
\centering\setlength{\tabcolsep}{5pt}
\begin{minipage}{0.49\columnwidth}\centering
\begin{tabular}{@{}lll@{}}
\toprule \tx{email} & \tx{invoice} & \tx{overdue} \\ \midrule
\tx{a@x} & I-1 & yes \\ \tx{b@x} & I-2 & no \\
\bottomrule
\end{tabular}\\[1pt]\textsc{Billing}
\end{minipage}%
\hfill
\begin{minipage}{0.49\columnwidth}\centering
\begin{tabular}{@{}lll@{}}
\toprule \tx{cid} & \tx{ticket} & \tx{severity} \\ \midrule
C1 & T-1 & hi \\ C2 & T-2 & lo \\
\bottomrule
\end{tabular}\\[1pt]\textsc{Support}
\end{minipage}\\[6pt]
\fbox{%
\begin{minipage}{0.90\columnwidth}\centering
\textsc{HR} \emph{resolver} (optional)---interface law $\tx{email}\to\tx{cid}$:\quad $\tx{a@x}\mapsto\mathrm{C1}$,\; $\tx{b@x}\mapsto\mathrm{C2}$
\end{minipage}}
\end{table}

\begin{table}[t]
\caption{Two worlds that agree on the observed \textsc{Billing} and \textsc{Support} views of Table~\ref{tab:re-schema} but link \tx{email} to \tx{cid} oppositely. Under the resolver law $\tx{email}\to\tx{cid}$, world $w$ is legal and $w'$ is not; without the law both are legal.}
\label{tab:two-worlds}
\centering
\begin{tabular}{@{}llllll@{}}
\toprule
\tx{email} & \tx{invoice} & \tx{overdue} & \tx{cid} & \tx{ticket} & \tx{severity} \\
\midrule
\tx{a@x} & I-1 & yes & C1 & T-1 & hi \\
\tx{b@x} & I-2 & no  & C2 & T-2 & lo \\
\bottomrule
\end{tabular}\\[2pt]
\emph{world $w$}---legal under $\tx{email}\to\tx{cid}$\par\vspace{6pt}
\begin{tabular}{@{}llllll@{}}
\toprule
\tx{email} & \tx{invoice} & \tx{overdue} & \tx{cid} & \tx{ticket} & \tx{severity} \\
\midrule
\tx{a@x} & I-1 & yes & C2 & T-2 & lo \\
\tx{b@x} & I-2 & no  & C1 & T-1 & hi \\
\bottomrule
\end{tabular}\\[2pt]
\emph{world $w'$}---illegal: \tx{a@x} linked to \tx{cid}{=}C2
\end{table}

The running example defines multiple attributes (e.g., \tx{email} and \tx{severity}). The \emph{attribute universe} $\AttrUniverse$ is the finite set of all the attributes. A \emph{world} $w$ is a finite relation over $\AttrUniverse$, and a \emph{view} $w|_O$ is its projection to a subset $O\subseteq\AttrUniverse$. If $R$ denotes the view $w|_O$ then we write $\attribute{R}$ for the set $O$. Table~\ref{tab:two-worlds} shows two such worlds, $w$ and $w'$: they agree on the views \textsc{Billing} and \textsc{Support} yet describe different underlying data. So, views alone do not single out one world. \emph{Interface laws} are the interface's mechanism for eliminating this ambiguity. Given two sets of attributes $X,Y \subseteq \AttrUniverse$, an \emph{interface law} is an expression of the form $X \to Y$. The running example has a single interface law, the resolver law $\mathrm{HR} = \{\tx{email}\} \to \{\tx{cid}\}$. An \emph{interface} is a pair $(\LawSet,\OverlapSet)$ where $\LawSet$ is a finite set of interface laws, and $\OverlapSet$ is a finite set of subsets of $\AttrUniverse$, called \emph{designated overlaps}. In the example, $\LawSet = \{\mathrm{HR}\}$ and $\Omega = \{\attribute{\textsc{Billing}}, \attribute{\textsc{Support}}\}$.

So far, we have described the syntax of interfaces and worlds. Definition~\ref{def:legality} defines its semantics. That is, which worlds are legal. The resolver law is what separates $w$ from $w'$.

\begin{definition}\label{def:legality}
A \emph{legality structure} $(\LegalWorlds,\LawSet)$ consists of a nonempty set of worlds $\LegalWorlds$ and a set of functional dependencies $\LawSet$ such that, for every $X\to Y\in\LawSet$ and all $w,w'\in\LegalWorlds$, $s\in w$, $t\in w'$,
\[s|_X = t|_X \;\Longrightarrow\; s|_Y = t|_Y.\]
The worlds in $\LegalWorlds$ are \emph{legal}.
\end{definition}

\begin{remark}
  An interface law is stronger than a single-instance \fd: it constrains tuples drawn from \emph{different} worlds, modeling a shared, fixed component---\eg a resolver mapping one identifier to another---whose behaviour is identical in every legal world.
\end{remark}

\begin{example}\label{ex:legality}
Take $\LawSet=\{\tx{email}\to\tx{cid}\}$, with the resolver of Table~\ref{tab:re-schema} fixing $\tx{a@x}\mapsto\mathrm{C1}$, $\tx{b@x}\mapsto\mathrm{C2}$. Of the two worlds in Table~\ref{tab:two-worlds}, $w$ is legal but $w'$ is not: pairing $w'$'s \tx{a@x} tuple ($\tx{cid}=\mathrm{C2}$) with $w$'s ($\tx{cid}=\mathrm{C1}$) gives equal \tx{email} but different \tx{cid}, violating $\tx{email}\to\tx{cid}$ across worlds. Both agree on the observed views; only $w$ also respects the resolver. Drop the resolver ($\LawSet=\varnothing$) and both become legal---the ambiguity behind non-identifiability.
\end{example}

The \emph{overlap evidence} is what the designated overlaps expose directly; together with everything the interface laws derive from it, this forms the \emph{interface evidence}.

How far does an interface law carry overlap evidence? In the running example, from \tx{email} to \tx{cid}; in general, the reach is captured by a single standard notion, \emph{attribute closure}, applied to the designated overlaps. Given $\LawSet$, the \emph{attribute closure} $\AttrClosure{X}$ of a set $X\subseteq\AttrUniverse$ is the standard Armstrong closure: the least superset of $X$ closed under all rules in $\LawSet$. It is computed in polynomial time via forward chaining.

\emph{Overlap augmentation} applies this closure to the overlaps: a \emph{designated overlap} $O\subseteq\AttrUniverse$ (one $O \in \Omega$) is extended to its \emph{closure-augmented schema} $\AugSchema{O} = \AttrClosure{O}$, the attributes the interface laws make determinable from $O$-evidence.

A single overlap closure captures what one overlap determines on its own; but a query may need attributes spread across several overlaps, linked only by chaining through the attributes those overlaps share. To track this combined reach---which attributes the interface ties together---we collect the per-overlap closures into one graph.

\begin{definition}\label{def:cc-graph}
The \emph{constraint-closed overlap graph} $\CCGraph$ has vertex set $\AttrUniverse$; for each designated overlap $O$, it adds a clique on $\AugSchema{O} = \AttrClosure{O}$. Its connected components are \emph{constraint-closed components}; $\closedComponent{a}$ denotes the component of $a\in\AttrUniverse$.
\end{definition}

\begin{example}\label{ex:cc-graph}
The two designated overlaps are the view schemas: $O_B = \{\tx{email},\tx{invoice},\tx{overdue}\}$ (\textsc{Billing}) and $O_S = \{\tx{cid},\tx{ticket},\tx{severity}\}$ (\textsc{Support}). With $\LawSet=\{\tx{email}\to\tx{cid}\}$: $\AugSchema{O_B} = \{\tx{email},\tx{invoice},\tx{overdue},\tx{cid}\}$ (the law extends the Billing closure to include \tx{cid}); $\AugSchema{O_S} = \{\tx{cid},\tx{ticket},\tx{severity}\}$. Both augmented overlaps contain \tx{cid}, so $\CCGraph$ has a single connected component spanning all six attributes. Without the law, $\AugSchema{O_B}=\{\tx{email},\tx{invoice},\tx{overdue}\}$ and $\AugSchema{O_S}=\{\tx{cid},\tx{ticket},\tx{severity}\}$ share no attribute---two disjoint components, one per view.
\end{example}

Given an interface $(\LawSet,\OverlapSet)$, the graph is computed in polynomial time in $|\AttrUniverse|+|\LawSet|+|\text{overlaps}|$ by computing each $\AugSchema{O}$ and then taking connected components.

\section{Identifiability}
\label{sec:identifiability}
Section~\ref{sec:model} fixed what an interface exposes---its legal worlds and, through closure, the attributes each overlap determines. Whether this is enough to \emph{answer a query} is a separate question: the interface evidence may pin down $\query$'s answer across all legal worlds, or leave it ambiguous. A query of the first kind is \emph{identifiable}; this section makes the notion precise, characterizes when it holds, and quantifies what is lost when it does not.

\subsection{Observational equivalence and identifiability}
The interface reveals a world only through its closure-augmented overlap projections. Two different worlds can project identically onto every overlap, leaving them indistinguishable from the interface evidence alone. We capture this as \emph{observational equivalence}, and call a query \emph{identifiable} when its answer never differs between two such worlds.

\begin{definition}\label{def:obs-equiv}
We call $w$ and $w'$ \emph{obs-equivalent}, written $w\sim w'$, if $w|_{\AugSchema{O}} = w'|_{\AugSchema{O}}$ for every designated overlap $O$ (where $\AugSchema{O}=\AttrClosure{O}$). Equivalently, writing $\mathrm{Obs}(w)=(w|_{\AugSchema{O}})_{O\in\Omega}$ for the \emph{observation} of $w$, we have $w\sim w'$ iff $\mathrm{Obs}(w)=\mathrm{Obs}(w')$.
\end{definition}

\begin{example}\label{ex:obs-equiv}
Continue the running example of Table~\ref{tab:re-schema}: its two designated overlaps are the \textsc{Billing} and \textsc{Support} view schemas, $O_B$ and $O_S$. Obs-equivalence requires agreement on their closure-augmented schemas, $\AugSchema{O_B}=\{\tx{email},\allowbreak\tx{invoice},\allowbreak\tx{overdue},\allowbreak\tx{cid}\}$ (the resolver law $\tx{email}\to\tx{cid}$ adds \tx{cid}) and $\AugSchema{O_S}=\{\tx{cid},\allowbreak\tx{ticket},\allowbreak\tx{severity}\}$. The worlds $w$ and $w'$ of Table~\ref{tab:two-worlds} make this concrete: they agree on the \textsc{Billing} and \textsc{Support} instances but assign \tx{a@x} different \tx{cid} values, so they disagree on $\AugSchema{O_B}$ (which contains \tx{cid}) and are \emph{not} obs-equivalent. Without the resolver law, \tx{cid} leaves $\AugSchema{O_B}$ and the two become obs-equivalent, differing only in a linkage the interface no longer determines.
\end{example}

Identifiability asks whether obs-equivalence forces a query's answer to agree. We take queries to be \emph{conjunctive}, the standard language of joins and projections: a \cq{} $\query=\exists\bar z\,\bigwedge_j R_{U_j}(\bar v_j)$ references the attributes $\attribute{\query}=\bigcup_j U_j$, its \emph{footprint}, and a Boolean \cq{} returns true or false on each world.

\begin{definition}\label{def:query-ident}
A \cq{} $\query$ is \emph{identifiable} from the interface evidence if $\query(w)=\query(w')$ for all legal worlds $w,w'\in\LegalWorlds$ with $w\sim w'$.
\end{definition}

\begin{example}\label{ex:identifiable}
With $\LawSet=\{\tx{email}\to\tx{cid}\}$, the query $\query$ reads \tx{overdue} from $R_{\AugSchema{O_B}}$ and \tx{severity} from $R_{\AugSchema{O_S}}$, joined on \tx{cid}~$\in\AugSchema{O_B}$: every atom uses a symbol from $\mathcal{L}_{\mathrm{iv}}$, so Theorem~\ref{thm:iv-cq-iff} certifies $\query$ as identifiable. Without the law ($\LawSet=\varnothing$): $\AugSchema{O_B}=\{\tx{email},\tx{invoice},\tx{overdue}\}$ loses \tx{cid}, leaving \textsc{Billing} and \textsc{Support} unlinked. The two worlds $w$ and $w'$ of Table~\ref{tab:two-worlds} then agree on both view projections yet answer $\query$ differently: in $w$, overdue \tx{a@x} is paired with high-severity \tx{cid}{=}C1 ($\query=\mathrm{true}$); in $w'$, with low-severity C2 ($\query=\mathrm{false}$). Same evidence, different answers---$\query$ is not identifiable without the resolver law.
\end{example}

\subsection{Certifying identifiability}
Identifiability quantifies over all obs-equivalent pairs of legal worlds---far too many to test directly. But two structural conditions, each decidable from the schema alone, are sufficient. The first asks that the query's footprint be covered by a single overlap closure:

\begin{theorem}[Closure certificate]\label{thm:frontier-suff}
Let $\query$ be a \cq. If there exists a designated overlap $O$ such that
\[\attribute{\query} \;\subseteq\; \AttrClosure{O},\]
then $\query$ is identifiable.
\end{theorem}

\begin{proof}
By Definition~\ref{def:obs-equiv}, $w\sim w'$ implies $w|_{\AugSchema{O}}=w'|_{\AugSchema{O}}$ for every designated overlap $O$. If $\attribute{\query}\subseteq\AttrClosure{O}=\AugSchema{O}$, both worlds agree on every attribute in $\attribute{\query}$. Equal footprint projections mean each atom $R_{U_j}$ of $\query$ has identical extension in $w$ and $w'$; therefore $\query(w)=\query(w')$.
\end{proof}

\paragraph{Interface-visible queries.}
Theorem~\ref{thm:frontier-suff} certifies a \cq{} by inspecting its footprint $\attribute{\query}$ and checking that it falls within a single overlap closure---a test repeated for each query. For a broad class of queries this per-query check is unnecessary: those \emph{written over} the augmented overlap layer are identifiable by their vocabulary alone.

For each designated overlap $O$, let $R_{\AugSchema{O}}$ be a relation symbol of arity $|\AugSchema{O}|$, interpreted on world $w$ as $w|_{\AugSchema{O}}$, the projection of $w$ to the closure-augmented overlap schema. The \emph{interface-visible vocabulary} is $\mathcal{L}_{\mathrm{iv}}=\{R_{\AugSchema{O}}\mid O\in\OverlapSet\}$, and a \cq{} is \emph{interface-visible} when every one of its atoms uses a symbol from $\mathcal{L}_{\mathrm{iv}}$.

Multi-view pretraining objectives are a canonical instance: contrastive loss and co-occurrence prediction compute functions of the form $\ell(w|_{\AugSchema{O}},w'|_{\AugSchema{O}})$, which are \cq{s} over $\mathcal{L}_{\mathrm{iv}}$. The next result shows every such query is identifiable, with no footprint inspection at all.

\begin{theorem}[Interface-visible identifiability]\label{thm:iv-cq-iff}
Every interface-visible \cq\ is identifiable.
\end{theorem}

\begin{proof}
Let $w\sim w'$. By Definition~\ref{def:obs-equiv}, $w|_{\AugSchema{O}}=w'|_{\AugSchema{O}}$ for every $O\in\OverlapSet$, so every symbol in $\mathcal{L}_{\mathrm{iv}}$ has identical extension (as a set of tuples) in both worlds. Since $\query$ uses only symbols from $\mathcal{L}_{\mathrm{iv}}$, its evaluation depends entirely on these extensions, and therefore $\query(w)=\query(w')$.
\end{proof}

Theorems~\ref{thm:frontier-suff} and~\ref{thm:iv-cq-iff} are complementary. Theorem~\ref{thm:frontier-suff} certifies any \cq\ whose footprint happens to fall in a single closure, regardless of how the query is written; Theorem~\ref{thm:iv-cq-iff} certifies \cq{s} \emph{written over} $\mathcal{L}_{\mathrm{iv}}$, including joins across different overlaps, with no closure condition. Multi-view pretraining objectives therefore sit inside the identifiable layer by their vocabulary alone; the contribution of this paper is to characterise which downstream queries fall \emph{outside} it.

\subsection{The cost of non-identifiability}
The certificates above say when a query \emph{is} identifiable. But when a query fails them, does that cost anything---could more data or a larger model still recover the answer? It cannot: the obstruction is structural, not statistical, and the only remedy is to change the interface. We begin with the error floor:

\begin{theorem}[Minimax lower bound]\label{thm:nonident-minimax}
If $\query$ is not identifiable, there exist $w\sim w'$ with $\query(w)\neq \query(w')$. For any estimator $\widehat{\query}$ whose output depends only on interface evidence,
\[\sup_{u\in\{w,w'\}} \Pr\!\left[\widehat{\query}(u)\neq \query(u)\right] \;\geq\; \tfrac{1}{2}\]
under $0$--$1$ loss. The $\tfrac{1}{2}$ error floor is irreducible.
\end{theorem}

\begin{proof}
Since $\query$ is not identifiable, by Definition~\ref{def:query-ident} there exist $w\sim w'$ with $w,w'\in\LegalWorlds$ and $\query(w)\neq \query(w')$. Since $w\sim w'$, both worlds produce identical interface evidence; any estimator $\widehat{\query}$ depending only on interface evidence satisfies $\widehat{\query}(w)=\widehat{\query}(w')$.

Without loss of generality, $\query(w)=\mathrm{true}$ and $\query(w')=\mathrm{false}$. If $\widehat{\query}(w)=\widehat{\query}(w')=\mathrm{true}$, then $\widehat{\query}$ errs on $w'$; if $\widehat{\query}(w)=\widehat{\query}(w')=\mathrm{false}$, it errs on $w$. In both cases the estimator misclassifies at least one element of $\{w,w'\}$, so
\[\sup_{u\in\{w,w'\}} \Pr\!\left[\widehat{\query}(u)\neq \query(u)\right] \;\geq\; \tfrac{1}{2}.\]
No estimator depending only on interface evidence can do better, since it cannot distinguish $w$ from $w'$.
\end{proof}

\paragraph{Outcome lower bounds.}
Theorem~\ref{thm:nonident-minimax} establishes that no predictor can beat the $\nicefrac{1}{2}$ error floor on non-identifiable queries. But that bound is uniform---it says nothing about the cost of correctly answering queries that \emph{are} identifiable. For those, the interface imposes a dual constraint that Theorem~\ref{thm:rate-lb} makes precise: a minimum number of bits any correct predictor must encode.

Write $m_{\query} = |\{\query(w) \mid w\in\LegalWorlds\}|$ for the \emph{outcome multiplicity} of an identifiable \cq{} $\query$---the number of distinct answer values it realises over all legal worlds.

\begin{theorem}[Outcome lower bound]\label{thm:rate-lb}
Let $\query$ be an identifiable \cq with multiplicity $m_{\query}$. Any predictor that reads interface evidence, stores it in a representation of size at most $2^k$, and answers $\query$ correctly on every $w\in\LegalWorlds$ must satisfy $k\ge\log_2 m_{\query}$.
\end{theorem}

\begin{proof}
Since $\query$ is identifiable it induces a well-defined map on $\LegalWorlds/{\sim}$. Distinct elements of $Y_{\query}=\{\query(w)\mid w\in\LegalWorlds\}$ correspond to distinct equivalence classes, so the representation must have at least $m_{\query}$ states: $2^k\ge m_{\query}$.
\end{proof}

The distributional version follows from Fano's inequality~\cite{fano1961transmission}:

\begin{corollary}[Fano lower bound]\label{cor:fano}
Let $\query$ be an identifiable \cq with $m_{\query}\ge 2$. Under a uniform prior over obs-equivalence classes, any predictor whose internal representation $R$ has at most $2^k$ states satisfies
\[
\Pr[\widehat{\query}\neq \query]\;\ge\;1-\frac{I(\query;R)+1}{\log_2 m_{\query}}\;\ge\;1-\frac{k+1}{\log_2 m_{\query}}.
\]
Achieving error at most $\delta$ requires $k\ge(1-\delta)\log_2 m_{\query} - 1$.
\end{corollary}

For Boolean queries ($m_{\query}\le 2$) a single bit suffices; the bound becomes strictly informative for $m_{\query}>2$, where the required capacity grows as $\log_2 m_{\query}$---quantifying the minimum interface complexity of answering $\query$.

The results so far treat the interface as fixed, characterising which queries are identifiable under a given $(\LawSet,\Omega)$. But the interface is itself a design target---adding overlaps or laws changes it. Even a single such addition can flip a query's status discontinuously, a \emph{capability jump}:

\begin{theorem}[Capability jumps]\label{thm:connectivity-emergence}
Let $\query$ have atoms $R_{U_1},\dots,R_{U_m}$. Under an augmented interface with law set $\LawSet'$ and overlaps $\Omega'$: if for every atom $j$ there exists a designated overlap $O_j$ (under $\Omega'$) with $\attribute{U_j}\subseteq \AttrClosure[\LawSet']{O_j}$, then $\query$ is identifiable under the augmented interface. If $\query$ was not identifiable before augmentation, this constitutes a structural capability jump.
\end{theorem}

\begin{proof}
By Definition~\ref{def:obs-equiv}, $w\sim' w'$ under $(\LawSet',\Omega')$ requires $w|_{\AttrClosure[\LawSet']{O_j}} = w'|_{\AttrClosure[\LawSet']{O_j}}$ for every designated overlap $O_j\in\Omega'$. Since $\attribute{U_j}\subseteq\AttrClosure[\LawSet']{O_j}$ for each atom $j$, both worlds agree on $\attribute{U_j}$, and therefore $\query(w)=\query(w')$.
\end{proof}

\subsection{Identifiability and multi-view pretraining}
Identifiability is a structural property of the interface. But a multi-view pretraining system never computes a closure---it only drives down an overlap-anchored loss. Does the structural notion constrain what such a system can deliver? It does: once the loss falls below a fixed threshold, the predictor must return the unique answer on every certified query, regardless of capacity or sample size. We measure prediction quality by Jensen--Shannon divergence~\cite{lin2002divergence}, under which the identifiability dichotomy induces a sharp quantitative boundary.

\begin{definition}\label{def:pair-disc}
The \emph{pair discrepancy} of $\query$ under interface $I$ is
\[
\Delta_{\mathrm{JS}}(\query,I)=\sup\bigl\{
  \mathrm{JS}(\delta_{\query(w)}\|\delta_{\query(w')})
  :w,w'\in\LegalWorlds,\; w\sim w'
\bigr\}.
\]
\end{definition}

Since $\query$ takes values in a finite set, $\mathrm{JS}(\delta_{\query(w)}\|\delta_{\query(w')})$ equals $1$~bit when $\query(w)\neq \query(w')$ and $0$ otherwise (using $\log_2$; this holds for all \cq{}s: any two distinct answer values, whether Boolean or tuple-sets, yield $\mathrm{JS}=1$~bit between the corresponding Dirac masses); hence $\Delta_{\mathrm{JS}}(\query,I)\in\{0,1\}$.

\begin{theorem}[Zero-discrepancy threshold]\label{thm:js-threshold}
$\Delta_{\mathrm{JS}}(\query,I)=0$ if and only if $\query$ is identifiable under $I$. Consequently, if $\attribute{\query}\subseteq\AttrClosure{O}$ for some $O\in\OverlapSet$, the interface achieves $\Delta_{\mathrm{JS}}=0$; augmenting the interface to satisfy this condition reduces pair discrepancy discontinuously from $1$ to $0$.
\end{theorem}

\begin{proof}
($\Rightarrow$) If $\query$ is identifiable then $\query(w)=\query(w')$ for every $w\sim w'$, so every term in the sup is $0$.
($\Leftarrow$) Contrapositive: if $\query$ is not identifiable there exist $w\sim w'$ with $\query(w)\neq \query(w')$, giving $\Delta_{\mathrm{JS}}=1>0$.  The closure certificate (Theorem~\ref{thm:frontier-suff}) provides a checkable sufficient condition for $\Delta_{\mathrm{JS}}=0$.
\end{proof}

The threshold is sharp: no augmentation can achieve $0<\Delta_{\mathrm{JS}}<1$.  A quantitative, noise-tolerant refinement---a \emph{small} overlap loss, not only an exactly zero one, already forces exact agreement on certified queries---is given below (Theorem~\ref{thm:robust-threshold}).  Sections~\ref{sec:exp-rq1}--\ref{sec:exp-scalability} validate Theorems~\ref{thm:frontier-suff}--\ref{thm:connectivity-emergence} directly.

\paragraph{Implication for pretraining practice.}
Theorem~\ref{thm:js-threshold} separates two distinct quantities.  The pair discrepancy $\Delta_{\mathrm{JS}}(\query,I)\in\{0,1\}$ is \emph{structural}: determined by the interface alone, $\Delta_{\mathrm{JS}}=0$ iff $\query$ is identifiable.  The training loss is a \emph{model} quantity: a pretraining pipeline drives its overlap loss toward zero, and a predictor that lands in the low-loss slice is constrained by Theorem~\ref{thm:robust-threshold} to return the unique answer on every certified query, independently of sample size.  For non-certified queries, $\Delta_{\mathrm{JS}}=1$ is irreducible regardless of training, and the $\nicefrac{1}{2}$ floor of Theorem~\ref{thm:nonident-minimax} applies.  Certifiability is therefore the checkable structural condition separating what pretraining \emph{can} deliver from what it fundamentally \emph{cannot}; the MinAug prescriptions in \S\ref{sec:algorithms} identify the minimal interface changes that move a non-certifiable query across that boundary.

\paragraph{Pretraining objective and noise tolerance.}
We make the pretraining connection precise. A multi-view model is an amortised
inference map $q_\theta(\cdot\mid E)\in\Delta(\LegalWorlds)$ from interface evidence
$E$ to a posterior over legal worlds, evaluated through the induced query estimator
$\widehat{\query}(E)$. Pretraining minimises an overlap-anchored loss whose term for a
designated overlap $O$ compares each world's closure-augmented projection to an
\emph{anchor} $p_O$ fixed across worlds (\eg the observed evidence on $\AugSchema{O}$):
\[
  \ell_{\ov}(w)\;\ge\;\eta_O\,\JS\!\bigl(\delta_{w|_{\AugSchema{O}}}\,\big\|\,p_O\bigr),
  \qquad \eta_O>0,
\]
with $\JS$ in nats ($\kappa{=}1$). The zero-discrepancy threshold then admits a
quantitative form: a \emph{small} loss already forces exact agreement.

\begin{lemma}\label{lem:js-mode}
Let $p\in\Delta(\ma X)$ on a finite set and $x,x'\in\ma X$. If $\JS(\delta_x\|p)\le\gamma$
and $\JS(\delta_{x'}\|p)\le\gamma$ with $\gamma<\tfrac{1}{8\kappa}$, then $x=x'$.
\end{lemma}

\begin{proof}
For a point mass, a Pinsker-type bound for Jensen--Shannon divergence
\cite{lin2002divergence} gives $p(x)\ge 1-\sqrt{2\kappa\gamma}$ whenever
$\JS(\delta_x\|p)\le\gamma$; for $\gamma<\tfrac1{8\kappa}$ this exceeds $\tfrac12$, and
likewise $p(x')>\tfrac12$. A finite distribution has at most one value of mass
exceeding $\tfrac12$, so $x=x'$.
\end{proof}

Call $\query$ \emph{$(\varepsilon,0)$-identifiable} when $\query(w)=\query(w')$ for all $w,w'\in\LegalWorlds$ with $\ell_{\ov}(w),\ell_{\ov}(w')\le\varepsilon$. A small enough loss guarantees it:

\begin{theorem}[Robust threshold]\label{thm:robust-threshold}
Suppose $\attribute{\query}\subseteq\AttrClosure{O}$ for a designated overlap $O$ whose loss
term is anchored as above, and let $\varepsilon_0=\eta_O/(8\kappa)$. Then for every
$\varepsilon<\varepsilon_0$, $\query$ is $(\varepsilon,0)$-identifiable.
\end{theorem}

\begin{proof}
If $\ell_{\ov}(w),\ell_{\ov}(w')\le\varepsilon<\varepsilon_0$, both anchored terms are at
most $\varepsilon/\eta_O<\tfrac1{8\kappa}$, so Lemma~\ref{lem:js-mode} gives
$w|_{\AugSchema{O}}=w'|_{\AugSchema{O}}$. Since
$\attribute{\query}\subseteq\AttrClosure{O}=\AugSchema{O}$, the two worlds agree on the whole
footprint; each atom $R_{U_j}$ has identical extension in both worlds, so $\query(w)=\query(w')$.
\end{proof}

\noindent A pretraining run that drives the overlap loss below $\varepsilon_0$ is therefore
forced to the unique answer on \emph{every} certified query---certified queries are
structurally determined, so this is independent of sample size; whether a given
architecture can drive loss below $\varepsilon_0$ is a separate architectural question---while
leaving the $\nicefrac{1}{2}$ floor of Theorem~\ref{thm:nonident-minimax} untouched for
non-certified queries. Certifiability, not data volume, governs what pretraining can deliver.

\subsection{Witnesses and completeness}
The closure certificate (Theorem~\ref{thm:frontier-suff}) is sufficient, but its failure is not a proof of non-identifiability: a query whose footprint escapes every overlap closure can still be identifiable if the legal class simply lacks the world pairs that would realise the gap the missing closure permits. The constructive dual settles it---a \emph{non-identifiability witness} for $\query$ under $I$ is a pair $(w,w')\in\LegalWorlds^2$ with $w\sim w'$ and $\query(w)\neq \query(w')$, and its existence proves non-identifiability directly (Proposition~\ref{prop:dual-cert}). On closure-separable instances every failed certificate yields such a witness, so the test is complete.

\begin{proposition}[Dual certificates]\label{prop:dual-cert}
\emph{(i)} $\query$ is identifiable iff no witness exists.
\emph{(ii)} If $\query$ satisfies the closure certificate of Theorem~\ref{thm:frontier-suff}---every atom footprint lies within some overlap closure---then no witness exists.
\emph{(iii)} If a witness exists, the closure certificate fails.
\end{proposition}

\begin{proof}
(i) is Definition~\ref{def:query-ident}.  (ii): the closure certificate gives $U_j\subseteq\AttrClosure{O_j}$ for each atom $R_{U_j}$ and some overlap $O_j$. By Definition~\ref{def:obs-equiv}, $w\sim w'$ implies $w|_{\AttrClosure{O_j}}=w'|_{\AttrClosure{O_j}}$ for each $j$; since $U_j\subseteq\AttrClosure{O_j}$, both worlds agree on every attribute in $U_j$, hence on relation $R_{U_j}$. Agreement on every atom relation gives $\query(w)=\query(w')$ (Theorem~\ref{thm:frontier-suff} is the special case where one overlap covers all atoms).  (iii) is the contrapositive of (ii).
\end{proof}

\begin{proposition}[Monotone witness shrinkage]\label{prop:witness-shrink}
Let $I'$ extend $I$ with additional overlaps or FDs.  The obs-equivalence $\sim'$ under $I'$ is finer than $\sim$: if $w\not\sim' w'$ then $(w,w')$ is no longer a witness under $I'$.  Augmentation can only eliminate witnesses, never create them.
\end{proposition}

\begin{proof}
Each new overlap or FD adds constraints to obs-equivalence, separating additional world pairs.  A pair that was indistinguishable under $I$ may be distinguishable under $I'$, removing it from the witness set.  No pair becomes indistinguishable by adding evidence, so $\{(w,w'):w\sim' w'\}\subseteq\{(w,w'):w\sim w'\}$.
\end{proof}

\paragraph{When is the certificate complete?}
Proposition~\ref{prop:dual-cert} shows certification is \emph{sound}. It is not complete
in general: a query can fail the closure test yet still be identifiable if no two legal
worlds happen to realise the ambiguity the missing closure permits. Completeness
requires the legal class to be \emph{closure-separable}, made precise below.

Say the interface has \emph{$S$-ambiguity}, for $S\subseteq\AttrUniverse$, if some $w\sim w'$ in $\LegalWorlds$ disagree on the projection, $w|_S\neq w'|_S$. This is exactly what makes a projection query fail:

\begin{theorem}[Projection witness]\label{thm:projection-witness}
Let $S\subseteq\attribute{U}$ for a view $U$, and suppose the interface has
$S$-ambiguity witnessed by $w\sim w'$. Then the projection \cq{}
$\query_S(\bar x_S)=\exists\bar z\,R_U(\bar v_U)$---with the positions in $S$ free and the
rest existentially quantified---is not identifiable.
\end{theorem}

\begin{proof}
As $w|_S\neq w'|_S$, some $S$-tuple lies in exactly one of $w|_S,w'|_S$, hence in the
answer set of $\query_S$ for exactly one of the two worlds; since $w\sim w'$,
Definition~\ref{def:query-ident} fails for $\query_S$.
\end{proof}

Call the legality structure \emph{closure-separable} for $S\subseteq\AttrUniverse$ when failure of closure coverage---$S\not\subseteq\AttrClosure{O}$ for every $O\in\OverlapSet$---forces $S$-ambiguity. On such instances the certificate is complete:

\begin{corollary}[Completeness on separable instances]\label{cor:closure-complete}
Under closure-separability for $S$, the closure certificate is \emph{complete} for the
projection \cq{} $\query_S$: $\query_S$ is identifiable if and only if $S\subseteq\AttrClosure{O}$
for some $O\in\OverlapSet$. Hence on closure-separable instances the closure certificate fails
only for genuinely non-identifiable queries, and the minimum augmentation is then tight: no
smaller set of interface actions can make $\query_S$ identifiable.
\end{corollary}

\begin{proof}
If such $O$ exists, Theorem~\ref{thm:frontier-suff} gives identifiability. Otherwise
closure-separability yields $w\sim w'$ with $w|_S\neq w'|_S$, and
Theorem~\ref{thm:projection-witness} produces a witness, so $\query_S$ is not identifiable.
\end{proof}

\begin{remark}
Closure-separability is a genuine richness condition: without it, failure of closure
coverage means \emph{uncertified}, not necessarily non-identifiable. The exhaustive
benchmark of \S\ref{sec:exp-rq1}---where every non-certified single-atom query admits an
explicit witness---is its empirical counterpart: those instances are closure-separable
by construction.
\end{remark}

A feasible \tx{MinAug} solution therefore eliminates every remaining witness (Proposition~\ref{prop:dual-cert}): the augmented interface certifies $\query$, so no witness pair can exist under it.  Finding the minimum such solution is the \tx{MinAug} problem (§\ref{sec:algorithms}).

\section{Algorithms}
\label{sec:algorithms}

\paragraph{\tx{CheckCert}.}
\tx{CheckCert} decides whether a given query $\query$ is certified by the interface $(\LawSet,\Omega)$. It computes $\AugSchema{O} = \AttrClosure{O}$ for each designated overlap $O$ via forward chaining, then checks each atom of $\query$ separately: atom $R_{U_j}$ is covered if $U_j\subseteq\AugSchema{O}$ for some designated overlap $O$, with a possibly different overlap per atom. If every atom is covered it returns \textsc{Certified}, otherwise \textsc{Uncertified} together with the uncovered atoms. Each covered atom is interface-visible, so certification implies identifiability by Theorem~\ref{thm:iv-cq-iff} (the special case in which one overlap covers the entire footprint is Theorem~\ref{thm:frontier-suff}); the test mirrors the atom obligations \tx{Greedy-MinAug} discharges, so it certifies multi-atom queries that join across different overlaps rather than requiring a single closure to cover the whole footprint. The computation is polynomial in $|\AttrUniverse|+|\LawSet|+|\text{overlaps}|$. For workloads, the overlap closures are precomputed once and shared across all queries.

\paragraph{Minimum interface augmentation.}
When $\query$ is not certified, the designer asks: what is the smallest set of new interface actions (\eg adding a resolver, identifier, or crosswalk) that would make $\query$ identifiable? An action $A\subseteq\AttrUniverse$ creates a new designated overlap whose closure-augmented schema is $\AttrClosure{A}$; its \emph{atom coverage} is the set of atoms of $\query$ it resolves:
\[\coverage(A)=\{\,j \mid \attribute{U_j}\subseteq \AttrClosure{A}\,\}.\]

\begin{definition}\label{def:minaug}
Given query $\query$ with atom obligations $\AtomObligs=\{1,\dots,m\}$ and candidate actions $\CandSet$, \emph{\tx{MinAug}} asks for a minimum-cardinality $\CandSet'\subseteq\CandSet$ with $\AtomObligs\subseteq\bigcup_{A\in\CandSet'}\coverage(A)$. For a workload, $\AtomObligs$ is the disjoint union of obligations over all queries.
\end{definition}

Any feasible \tx{MinAug} solution makes $\query$ identifiable: each chosen action becomes a designated overlap whose closure certifies the corresponding atom obligation, so Theorem~\ref{thm:connectivity-emergence} applies (each atom is covered by its own overlap under the augmented interface).

\tx{MinAug} reduces exactly to \textsc{Set Cover}. We use the standard greedy algorithm:

\begin{algorithm}[t]
\caption{\tx{Greedy-MinAug} (Unweighted \& Weighted)}
\label{alg:greedy-minaug}
\begin{algorithmic}[1]
\REQUIRE $\LawSet$, atom obligations $\AtomObligs$, atom schemas $\{\attribute{U_j}\}_{j\in\AtomObligs}$, candidates $\CandSet$, (optional) costs $c:\CandSet\to\mathbb{R}_{>0}$ (default $c\equiv 1$).
\ENSURE Selected actions $\GreedySol\subseteq\CandSet$.
\STATE \textbf{Precompute} $\coverage(A)\gets\{j\in\AtomObligs\mid\attribute{U_j}\subseteq \AttrClosure{A}\}$ for all $A\in\CandSet$.
\STATE $C\gets\varnothing$;\quad $\GreedySol\gets\varnothing$.
\WHILE{$C\neq\AtomObligs$}
  \STATE $A^*\gets\arg\max_{A\in\CandSet\setminus\GreedySol}|\coverage(A)\setminus C|/c(A)$
  \IF{$|\coverage(A^*)\setminus C|=0$} \STATE \textbf{return} \textsc{Infeasible} \ENDIF
  \STATE $\GreedySol\gets\GreedySol\cup\{A^*\}$;\quad $C\gets C\cup\coverage(A^*)$
\ENDWHILE
\STATE \RETURN $\GreedySol$
\end{algorithmic}
\end{algorithm}

\begin{theorem}\label{thm:minaug-hard-greedy}
The decision version of \tx{MinAug} is \textnormal{NP}-complete even when $\LawSet=\varnothing$. For Algorithm~\ref{alg:greedy-minaug} with $c\equiv 1$:
\[\size{\GreedySol} \;\leq\; H_{\size{\AtomObligs}}\cdot\size{\OptSol} \;\leq\; (1+\ln\size{\AtomObligs})\cdot\size{\OptSol},\]
where $\OptSol$ is an optimal solution and $H_k=\sum_{i=1}^k 1/i$. The same logarithmic bound holds for the weighted variant (replacing cardinality with total cost).
\end{theorem}

\begin{proof}
\emph{NP-hardness.}
Reduce from \textsc{Set Cover}: given universe $\{1,\dots,m\}$ and sets $T_1,\dots,T_p$, introduce one attribute $a_j$ per element, one unary atom $R_{\{a_j\}}$ per element as atom obligations $\AtomObligs=\{1,\dots,m\}$, and candidate action $A_i=\{a_j\mid j\in T_i\}$ for each set $T_i$, with $\LawSet=\varnothing$. Then $\AttrClosure{A_i}=A_i$ and $\coverage(A_i)=T_i$. A size-$k$ solution to \tx{MinAug} exists iff a size-$k$ set cover exists. Membership in NP is immediate: guess a subset of candidates and verify coverage in polynomial time.

\emph{Greedy bound.}
With $c\equiv 1$, Algorithm~\ref{alg:greedy-minaug} is the standard greedy \textsc{Set Cover} algorithm applied to the coverage function. The $H_{\size{\AtomObligs}}$ bound is the classical greedy approximation ratio for \textsc{Set Cover}; the weighted variant follows from the analogous weighted analysis.
\end{proof}

\section{Experimental Evaluation}
\label{sec:experiments}

\paragraph{Schema-design walkthrough.}
Theorems~\ref{thm:frontier-suff}--\ref{thm:connectivity-emergence} translate directly into a two-step schema-design workflow. \emph{Step~1 (CheckCert):} given the current interface $(\LawSet,\Omega)$ and a target query $\query$, run \tx{CheckCert} to decide whether every atom footprint of $\query$ lies within some overlap closure $\AttrClosure{O}$ (a possibly different overlap per atom). If so, $\query$ is identifiable and no augmentation is needed. \emph{Step~2 (Greedy-MinAug):} if $\query$ is not certified, enumerate candidate interface actions (\eg adding a resolver, crosswalk, or shared identifier), compute their atom coverages, and run \tx{Greedy-MinAug} to find the smallest set of actions that covers all atom obligations. Any feasible solution makes $\query$ identifiable by Theorem~\ref{thm:connectivity-emergence}. The experiments below validate both steps at scale.

We evaluate \tx{CheckCert} and \tx{Greedy-MinAug} as schema-design tools: given a relational schema with functional dependencies and an overlap policy, do the algorithms correctly decide identifiability, witness failures, and recommend minimal augmentations at practical cost?
We address four research questions.
\textbf{RQ1}~Is the closure certificate (Theorem~\ref{thm:frontier-suff}) exact for
single-atom Boolean \cq{s} under single-row world semantics---confirming every certified
query identifiable and finding an explicit witness for every non-certified one?
\textbf{RQ2}~Does the certificate correctly classify real integration datasets, and can non-identifiability be witnessed in real data?
\textbf{RQ3}~Does \tx{Greedy-MinAug} (Algorithm~\ref{alg:greedy-minaug}) achieve practical
approximation ratios on realistic schemas?
\textbf{RQ4}~Do \tx{CheckCert} and \tx{Greedy-MinAug} remain practical at database-scale schemas
($|\mathit{Attr}|$ and $|\LawSet|$ up to $10^3$)?
Section~\ref{sec:exp-ml} presents confirmatory ML classifier experiments showing that
non-identifiable queries exhibit the $\nicefrac{1}{2}$ error floor
(Theorem~\ref{thm:nonident-minimax}) and that augmentation produces the capability jumps
predicted by Theorem~\ref{thm:connectivity-emergence}.
All RQ1--RQ4 experiments run on CPU; the confirmatory ML experiments use a GPU node equipped with
an NVIDIA~A40 GPU (48\,GB VRAM) and a 128-core CPU, running CUDA~12.1 and PyTorch~2.3.

\subsection{Setup}
\label{sec:exp-setup}

\paragraph{Synthetic benchmark.}
We use a CRM-inspired schema with $n{=}5$ attributes over a binary domain ($d{=}2$) and three views
(customer, order, support), inducing overlaps on their pairwise intersections.
A legality structure $(\LegalWorlds,\LawSet)$ is instantiated by drawing functional dependencies uniformly
at random from attribute pairs, then generating worlds of $m$ tuples consistent with $\LawSet$ via
shared resolvers (lookup tables keyed on FD antecedents).
For each $(\LawSet,\Omega,\query)$, identifiability of the target \cq{} $\query$ is checked using
Theorem~\ref{thm:frontier-suff}.
\textit{Exactness benchmark (RQ1).} We use a resolver-free exhaustive enumerator. Worlds are all $d^n$ FD-satisfying single-row assignments ($m{=}1$), giving universal identifiability semantics (not restricted to resolver-generated worlds). For each $(\LawSet,\Omega,\query)$, the enumerator groups worlds by observation and checks whether obs-equivalent worlds agree on $\query$.

\textit{Confirmatory ML (Section~\ref{sec:exp-ml}).} Worlds are generated by the resolver model: $m\!\in\!\{10,30,50\}$ tuples consistent with $\LawSet$ via shared resolver tables (lookup tables keyed on FD antecedents); training-set sizes $N\!\in\!\{10^3,5{\times}10^3,5{\times}10^4\}$.

\paragraph{Real-world datasets.}
\emph{BibInteg} is built from the OpenAlex API~\cite{priem2022openalex}, using
$10{,}000$ computer-science papers (2015--2024) with verified DOIs.
Three views mirror the DBLP/ACM/SemanticScholar schema:
\emph{DBLP} exposes $\{\tx{title},\tx{author},\tx{year},\tx{venue}\}$;
\emph{ACM} exposes $\{\tx{title},\tx{author},\tx{year},\tx{doi}\}$;
\emph{SemanticScholar} exposes $\{\tx{title},\tx{author},\tx{year},\tx{n\_authors}\}$.
All three views share overlap $O{=}\{\tx{title},\tx{author},\tx{year}\}$.
Interface laws: $O\!\to\!\tx{venue}$, $O\!\to\!\tx{doi}$,
$O\!\to\!\tx{n\_authors}$, $\{\tx{year}\}\!\to\!\tx{decade}$, so
the $\LawSet$-closure of $O$ covers all seven attributes and every
single-view existential query is certified.
\emph{WDC-Product}~\cite{primpeli2019profiling} provides the schema for a design case study, drawing its attribute structure from the Web Data Commons product corpus spanning Amazon, Walmart, and Best Buy listings.
\emph{Amazon} exposes $\{\tx{brand},\tx{model},\tx{category},\tx{price},\tx{rating}\}$;
\emph{Walmart} exposes $\{\tx{brand},\tx{model},\tx{category},\tx{price},\tx{in\_stock}\}$;
\emph{Best Buy} exposes $\{\tx{brand},\tx{model},\tx{category},\tx{n\_reviews}\}$.
Shared overlap $O{=}\{\tx{brand},\tx{model},\tx{category}\}$;
interface laws $\{\tx{brand},\tx{model}\}\!\to\!\tx{category}$,
$\{\tx{brand},\tx{model}\}\!\to\!\tx{price}$,
$\{\tx{brand},\tx{model}\}\!\to\!\tx{in\_stock}$
yield $\AugSchema{O}{=}\{\tx{brand},\tx{model},\tx{category},\tx{price},\tx{in\_stock}\}$.
Both \tx{rating} and \tx{n\_reviews} are outside the closure, yielding three non-certified queries:
\texttt{Q\_highly\_rated} (Amazon \tx{rating}${\geq}4$), \texttt{Q\_reviewed} (Best Buy \tx{n\_reviews}${\geq}32$),
and \texttt{Q\_popular} (Amazon \tx{rating}${\geq}4$ joined with Best Buy \tx{n\_reviews}${\geq}32$).
Two certified queries ask whether a product is available (\texttt{Q\_available}: Walmart \tx{in\_stock}${=}1$, footprint $\{0,1,2,6\}\subseteq\AugSchema{O}$) or inexpensive (\texttt{Q\_cheap}: Walmart \tx{price}${<}p_0$, footprint $\{0,1,2,3\}\subseteq\AugSchema{O}$).
BibInteg and CrossKG-DBLP form single-tuple worlds ($m{=}1$); Amazon-Google and Fodors-Zagat form 2-tuple matched-pair worlds---all with far smaller world multiplicity than the confirmatory ML setup ($m\!\in\!\{10,30,50\}$).
\emph{WDC data construction.} Raw WDC listings lack \tx{rating}, \tx{n\_reviews}, and \tx{in\_stock}; these are synthesized from \tx{brand}/\tx{model} keys to enforce the certified/uncertified split by construction. WDC is a schema stress test; real interface-law validation is on the four real-record datasets (\S\ref{sec:exp-real}).
\emph{CrossKG-DBLP} aligns $11{,}800$ computer-science papers from DBLP and OpenAlex on shared overlap $O{=}\{\tx{title},\tx{author},\tx{year}\}$ with interface law $O\!\to\!\tx{doi}$. \texttt{Q\_publisher} (publication publisher, identified by DOI prefix) is certified; \texttt{Q\_large\_team} (co-author count exceeding a threshold) is uncertified---DBLP and OpenAlex maintain independent author lists with no reconciliation law across sources.
\emph{Amazon-Google}~\cite{primpeli2019profiling} aligns $1{,}046$ product pairs from Amazon and Google Shopping. Matched-pair identity forms the shared overlap~$O$; \texttt{Q\_catalog} (whether a product is in a given catalog segment) is certified. No interface law constrains \tx{price} across sources---Amazon and Google list different prices for the same product (e.g., \$395 vs.\ \$319.95)---so \texttt{Q\_expensive} (\tx{price}${\geq}\$50$) is uncertified.
\emph{Fodors-Zagat}~\cite{primpeli2019profiling} aligns $110$ restaurant pairs from the Fodors and Zagat guides. Matched-pair identity and shared restaurant segment form~$O$; \texttt{Q\_segment} is certified. No interface law reconciles cuisine labels---the two guides categorize the same restaurant differently (e.g., ``asian'' vs.\ ``japanese'')---so \texttt{Q\_cuisine} is uncertified. This dataset is inherently small (${\sim}112$ gold pairs total in the Magellan benchmark); we include it for domain breadth.
Tables~\ref{tab:bibinteg-schema} and~\ref{tab:wdc-schema} give the attribute schemas, interface laws, and query certification status for BibInteg and WDC-Product.

\begin{table}[t]
  \centering
  \caption{\emph{BibInteg} schema.
    $O{=}\{0,1,2\}$;
    FDs: $O{\to}3$, $O{\to}4$, $O{\to}5$, $\{2\}{\to}6$;
    $\AugSchema{O}{=}\{0,\ldots,6\}$. All queries certified (\textbf{C}).}
  \label{tab:bibinteg-schema}
  \begin{tabular}{clccccc}
    \toprule
    Attr & Semantics & DBLP & ACM & SS & ${\in}O$ & ${\in}\AugSchema{O}$ \\
    \midrule
    0 & \texttt{title}    & \checkmark & \checkmark & \checkmark & \checkmark & \checkmark \\
    1 & \texttt{author}   & \checkmark & \checkmark & \checkmark & \checkmark & \checkmark \\
    2 & \texttt{year}     & \checkmark & \checkmark & \checkmark & \checkmark & \checkmark \\
    3 & \texttt{venue}    & \checkmark & --         & --         & --         & \checkmark \\
    4 & \texttt{doi}      & --         & \checkmark & --         & --         & \checkmark \\
    5 & \texttt{n\_authors}& --         & --         & \checkmark & --         & \checkmark \\
    6 & \texttt{decade}   & --         & --         & --         & --         & \checkmark \\
    \midrule
    \multicolumn{7}{l}{\texttt{Q\_venue}: footprint $\{0,1,2,3\}\subseteq\AugSchema{O}$\quad\textbf{C}} \\
    \multicolumn{7}{l}{\texttt{Q\_doi}: footprint $\{0,1,2,4\}\subseteq\AugSchema{O}$\quad\textbf{C}} \\
    \multicolumn{7}{l}{\texttt{Q\_large\_team}: footprint $\{0,1,2,5\}\subseteq\AugSchema{O}$\quad\textbf{C}} \\
    \bottomrule
  \end{tabular}
\end{table}

\begin{table}[t]
  \centering
  \caption{\emph{WDC-Product} schema.
    $O{=}\{0,1,2\}$;
    FDs: $\{0,1\}{\to}2$, $\{0,1\}{\to}3$, $\{0,1\}{\to}6$;
    $\AugSchema{O}{=}\{0,1,2,3,6\}$.
    Attrs~4, 5 lie outside $\AugSchema{O}$; three queries uncertified (\textbf{U}).}
  \label{tab:wdc-schema}
  \begin{tabular}{clccccc}
    \toprule
    Attr & Semantics & Amazon & Walmart & BestBuy & ${\in}O$ & ${\in}\AugSchema{O}$ \\
    \midrule
    0 & \texttt{brand}     & \checkmark & \checkmark & \checkmark & \checkmark & \checkmark \\
    1 & \texttt{model}     & \checkmark & \checkmark & \checkmark & \checkmark & \checkmark \\
    2 & \texttt{category}  & \checkmark & \checkmark & \checkmark & \checkmark & \checkmark \\
    3 & \texttt{price}     & \checkmark & \checkmark & --         & --         & \checkmark \\
    4 & \texttt{rating}    & \checkmark & --         & --         & --         & --         \\
    5 & \texttt{n\_reviews}& --         & --         & \checkmark & --         & --         \\
    6 & \texttt{in\_stock} & --         & \checkmark & --         & --         & \checkmark \\
    \midrule
    \multicolumn{7}{l}{\texttt{Q\_available}: footprint $\{0,1,2,6\}\subseteq\AugSchema{O}$\quad\textbf{C}} \\
    \multicolumn{7}{l}{\texttt{Q\_cheap}: footprint $\{0,1,2,3\}\subseteq\AugSchema{O}$\quad\textbf{C}} \\
    \multicolumn{7}{l}{\texttt{Q\_highly\_rated}: footprint $\{0,1,2,4\}\not\subseteq\AugSchema{O}$\quad\textbf{U}} \\
    \multicolumn{7}{l}{\texttt{Q\_reviewed}: footprint $\{0,1,2,5\}\not\subseteq\AugSchema{O}$\quad\textbf{U}} \\
    \multicolumn{7}{l}{\texttt{Q\_popular}: footprint $\{0,1,2,4,5\}\not\subseteq\AugSchema{O}$\quad\textbf{U}} \\
    \bottomrule
  \end{tabular}
\end{table}

\paragraph{Predictor architectures.}
We evaluate four architectures spanning the structure-agnostic to theory-exploiting spectrum, plus two baselines.
\emph{MLP}: view marginals are concatenated into a feature vector, passed through a two-layer ReLU
network with hidden dimension $h{=}64$.
\emph{SetTransformer}~(\tx{ST})~\cite{lee2019set}: attribute-value tokens from each view
form a set; cross-view multi-head attention precedes the classification head.
\emph{GNN-OG}: message passing over the constraint-closed overlap graph $\CCGraph$
(Definition~\ref{def:cc-graph}); overlap-marginal features propagate along overlap edges before pooling.
\emph{Closure-Aware}~(\tx{CA}): if the closure certificate holds, the answer is read
from the closure-augmented overlap projection without learned inference; otherwise falls back to
\tx{GNN-OG}.
Baselines: \emph{VanillaOv}~(logistic regression over $\LawSet$-closed overlap features; no
hidden layers, tests whether linear expressivity suffices) and
\emph{MajVote}~(constant majority predictor).

\paragraph{Metrics.}
Boolean \cq{s}: \emph{balanced accuracy} $\frac{1}{2}(\text{TPR}+\text{TNR})$, which equals
$\nicefrac{1}{2}$ for any majority-class predictor regardless of class imbalance, making it the
correct empirical proxy for the $\nicefrac{1}{2}$ error floor of Theorem~\ref{thm:nonident-minimax}.
\tx{MinAug}: approximation ratio $|\GreedySol|/|\OptSol|$ and wall-clock runtime.
All error bars are over worlds; results are averaged over three random seeds.

\subsection{Certificate Exactness (RQ1)}
\label{sec:exp-rq1}

We construct a benchmark of $841$ instances by sampling $200$ random schemas
($n\!\in\!\{4,6,8\}$ attributes, domain size $d{=}3$, up to four FDs) and generating five
single-atom Boolean \cq{s} per schema, split evenly between certified and non-certified.
For each instance we run both the closure certificate and an exhaustive identifiability verifier.
The verifier enumerates all $d^n$ single-row worlds under relational FD semantics---the semantics
of the theory, not the resolver model---and groups them by observation; for certified queries it
checks that no group contains conflicting answers, and for non-certified queries it returns an
explicit witness pair $(w, w')$ with $\mathrm{Obs}(w){=}\mathrm{Obs}(w')$ and $\query(w){\neq}\query(w')$.
Of the $353$ certified instances, all $353$ pass the observation-consistency check ($100\%$,
zero violations).
Of the $488$ non-certified instances, all $488$ have an explicit witness ($100\%$).
Within this restricted class (single-atom Boolean \cq{s}, single-row worlds), the certificate is exact on both sides; completeness is not claimed beyond this class.

\subsection{Real-World Certification and Witnesses (RQ2)}
\label{sec:exp-real}

\begin{table}[t]
  \centering
  \caption{Real-data coverage: 3 domains, 5 datasets, 14 queries.
    \textbf{C}~=~certified (Theorem~\ref{thm:frontier-suff});
    \textbf{W}~=~real witness found;
    \textbf{U}~=~uncertified, schema-only (WDC).}
  \label{tab:rq2-real}
  \begin{tabular}{@{}llc@{}}
    \toprule
    Dataset & Domain & C\,/\,W\,/\,U \\
    \midrule
    BibInteg (9,992 papers)        & scholarly  & 3C \\
    CrossKG-DBLP (11,800 papers)   & scholarly  & 1C\;1W \\
    Amazon-Google (1,046 pairs)    & product    & 1C\;1W \\
    Fodors-Zagat (110 pairs)       & restaurant & 1C\;1W \\
    WDC-Product (schema only)      & ---        & 2C\;3U \\
    \midrule
    Total                          &            & 8C\;3W\;3U \\
    \bottomrule
  \end{tabular}
\end{table}

\begin{table}[t]
  \caption{Balanced accuracy on all four real-world datasets. C\,$=$\,certified; U\,$=$\,uncertified; bold entries exceed the error floor ($> 0.505$); 3 seeds per cell; stds $= 0.00$ throughout.
  \tx{MLP}/\tx{ST}/\tx{VanillaOv}~$=$~$1.0$ on every certified query; \tx{GNN-OG}/\tx{CA}/\tx{MajVote}~$=$~$0.5$ on every certified query (see text).
  $N_{\mathrm{train}}{=}5000$ (BibInteg); ${\sim}120$ training pairs (Fodors-Zagat; \tx{VanillaOv}~$= 0.69$ is small-sample noise).}
  \label{tab:real-world}
  \centering
  \begin{tabular}{llcrrrrrr}
    \toprule
    Dataset & Query & C/U & MLP & ST & GNN-OG & CA & VanillaOv & MajVote \\
    \midrule
    \emph{BibInteg}      & \texttt{Q\_venue}       & C & \textbf{1.00} & \textbf{1.00} & 0.50 & 0.50 & \textbf{1.00} & 0.50 \\
                         & \texttt{Q\_doi}         & C & \textbf{1.00} & \textbf{1.00} & 0.50 & 0.50 & \textbf{1.00} & 0.50 \\
                         & \texttt{Q\_large\_team} & C & \textbf{1.00} & \textbf{1.00} & 0.50 & 0.50 & \textbf{1.00} & 0.50 \\
    \midrule
    \emph{CrossKG-DBLP}  & \texttt{Q\_publisher}   & C & \textbf{1.00} & \textbf{1.00} & 0.50 & 0.50 & \textbf{1.00} & 0.50 \\
                         & \texttt{Q\_large\_team} & U & 0.50 & 0.50 & 0.50 & 0.50 & 0.50 & 0.50 \\
    \midrule
    \emph{Amazon-Google} & \texttt{Q\_catalog}     & C & \textbf{1.00} & \textbf{1.00} & 0.50 & 0.50 & \textbf{1.00} & 0.50 \\
                         & \texttt{Q\_expensive}   & U & 0.50 & \textbf{0.52} & 0.50 & 0.50 & \textbf{0.51} & 0.50 \\
    \midrule
    \emph{Fodors-Zagat}  & \texttt{Q\_segment}     & C & \textbf{1.00} & \textbf{1.00} & 0.50 & 0.50 & \textbf{0.69} & 0.50 \\
                         & \texttt{Q\_cuisine}     & U & \textbf{0.52} & \textbf{0.54} & 0.50 & 0.50 & \textbf{0.53} & 0.50 \\
    \bottomrule
  \end{tabular}
\end{table}

\textbf{Answer to RQ2:}
Yes across all three domains (Table~\ref{tab:rq2-real}). On \emph{BibInteg}: all three queries are certified; laws hold on ${\geq}99.9\%$ of records, and the few violations are genuine non-identifiability witnesses (e.g., \emph{The ARRIVE guidelines~2.0} appears under one key with three venues and three DOIs). On \emph{CrossKG-DBLP}: \texttt{Q\_publisher} is certified; \texttt{Q\_large\_team} is uncertified---the same query is certified on BibInteg, whose law $O\!\to\!\texttt{n\_authors}$ reaches the author count, whereas DBLP and OpenAlex share no such reconciliation law---and a real witness is found in $1000/1000$ random trials (median discovery position $2{,}100$ records, positive rate $0.42$). On \emph{Amazon-Google}: \texttt{Q\_catalog} is certified; $9.4\%$ of matched pairs straddle the price threshold (first witness at scan position~62). On \emph{Fodors-Zagat}: \texttt{Q\_segment} is certified; cuisine labels disagree in ${\sim}80\%$ of matched pairs (first witness at scan position~2); the dataset is small and its value is domain breadth.

Table~\ref{tab:real-world} confirms the predicted pattern across all four datasets. \tx{MLP}, \tx{ST}, and \tx{VanillaOv} reach $1.0$ on every certified query; \tx{GNN-OG} and \tx{CA} remain at $0.50$ because the single-tuple and matched-pair world structures give their overlap-graph and closure-lookup components nothing to aggregate over (cf.\ Section~\ref{sec:exp-ml}, where both reach $1.0$ at $m{=}10$). Every uncertified query is at or within noise of the $0.50$ floor, confirming the minimax bound of Theorem~\ref{thm:nonident-minimax} on real data across three independent domains.

\subsection{Greedy-MinAug in Practice (RQ3)}
\label{sec:exp-rq3}

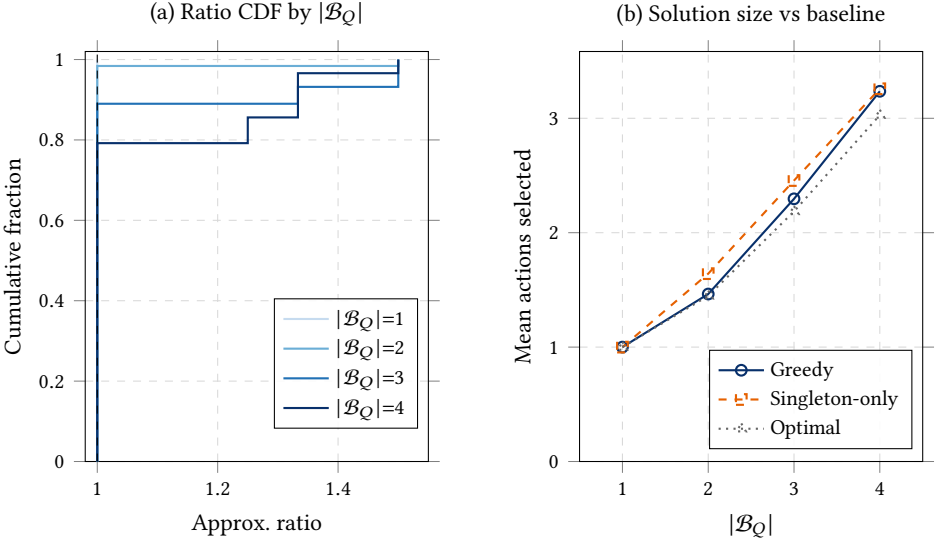
\begin{figure}[t]
  \centering
\definecolor{maB0}{RGB}{198,219,239}
\definecolor{maB1}{RGB}{107,174,214}
\definecolor{maB2}{RGB}{33,113,181}
\definecolor{maB3}{RGB}{8,48,107}
\definecolor{maGy}{RGB}{102,102,102}
\definecolor{maOr}{RGB}{230,97,1}
\begin{tikzpicture}
\begin{groupplot}[
  group style={group size=2 by 1, horizontal sep=2cm},
]
\nextgroupplot[
  width=0.44\textwidth,
  height=7cm,
  xlabel={Approx.\ ratio},
  ylabel={Cumulative fraction},
  xmin=0.98, xmax=1.55,
  ymin=0, ymax=1.02,
  xticklabel style={font=\footnotesize},
  yticklabel style={font=\footnotesize},
  xlabel style={font=\small},
  ylabel style={font=\small},
  legend style={at={(0.97,0.08)}, anchor=south east, font=\footnotesize,
                legend cell align=left, fill=white, fill opacity=0.9,
                draw opacity=1, text opacity=1},
  grid=major, grid style={dashed, gray!30},
  tick align=outside,
  title={(a) Ratio CDF by $|\mathcal{B}_Q|$},
  title style={font=\small},
]
\addplot[maB0, thick, const plot]
  table[x=ratio, y=cdf, col sep=comma]{pgfdata/minaug_cdf_k1.csv};
\addlegendentry{$|\mathcal{B}_Q|{=}1$}
\addplot[maB1, thick, const plot]
  table[x=ratio, y=cdf, col sep=comma]{pgfdata/minaug_cdf_k2.csv};
\addlegendentry{$|\mathcal{B}_Q|{=}2$}
\addplot[maB2, thick, const plot]
  table[x=ratio, y=cdf, col sep=comma]{pgfdata/minaug_cdf_k3.csv};
\addlegendentry{$|\mathcal{B}_Q|{=}3$}
\addplot[maB3, thick, const plot]
  table[x=ratio, y=cdf, col sep=comma]{pgfdata/minaug_cdf_k4.csv};
\addlegendentry{$|\mathcal{B}_Q|{=}4$}
\draw[black, thin, dashed] (axis cs:1.0,0) -- (axis cs:1.0,1.02);
\nextgroupplot[
  width=0.44\textwidth,
  height=7cm,
  xlabel={$|\mathcal{B}_Q|$},
  ylabel={Mean actions selected},
  xmin=0.5, xmax=4.5,
  xtick={1,2,3,4},
  ymin=0,
  xticklabel style={font=\footnotesize},
  yticklabel style={font=\footnotesize},
  xlabel style={font=\small},
  ylabel style={font=\small},
  legend style={at={(0.97,0.03)}, anchor=south east, font=\footnotesize,
                legend cell align=left, fill=white, fill opacity=0.9,
                draw opacity=1, text opacity=1},
  grid=major, grid style={dashed, gray!30},
  tick align=outside,
  title={(b) Solution size vs baseline},
  title style={font=\small},
]
\addplot[maB3, thick, mark=o, mark size=2pt]
  table[x=n_atoms, y=greedy_mean, col sep=comma]{pgfdata/minaug_sizes.csv};
\addlegendentry{Greedy}
\addplot[maOr, thick, mark=square, mark size=2pt, dashed]
  table[x=n_atoms, y=singleton_mean, col sep=comma]{pgfdata/minaug_sizes.csv};
\addlegendentry{Singleton-only}
\addplot[maGy, thick, mark=triangle, mark size=2pt, dotted]
  table[x=n_atoms, y=optimal_mean, col sep=comma]{pgfdata/minaug_sizes.csv};
\addlegendentry{Optimal}
\end{groupplot}
\end{tikzpicture}
  \caption{\tx{Greedy-MinAug} approximation ratio CDF~(a) and mean solution size~(b).
  Each curve in~(a) is one $|\AtomObligs|$ value; mass at $1.0$ indicates optimal.
  Panel~(b) compares \tx{Greedy} (with singleton and cross-atom pair candidates),
  \tx{Singleton-only} (pairs excluded), and the brute-force \tx{Optimal} per $|\AtomObligs|$.}
  \label{fig:minaug}
\end{figure}

We construct multi-atom \tx{MinAug} instances with $|\AtomObligs|\!\in\!\{1,2,3,4\}$ atom
obligations.
Each atom $U_i$ has $g{=}5$ dedicated attributes forming a chain
($\attribute{U_i} = \{5i, 5i{+}1, \ldots, 5i{+}4\}$ with within-atom FDs
$\{5i\}\!\to\!5i{+}1\!\to\!\cdots$), so a root singleton $\{5i\}$ covers all of atom $i$ via
closure.
A further $|\AtomObligs|$ random cross-atom FDs allow single actions to reach
multiple atoms transitively.
Candidates include the $|\AtomObligs|$ root singletons plus random cross-atom pairs
(total $\leq 16$ for brute-force feasibility; 500~trials per $|\AtomObligs|$; 2\,000~total).
The \emph{singleton-only} baseline restricts \tx{Greedy-MinAug} to root singletons only.

Figure~\ref{fig:minaug}\,(a) shows the approximation ratio CDF per $|\AtomObligs|$.
For $|\AtomObligs|{=}1$ the ratio is $1.000$ on every trial---consistent with the $H_1{=}1$
bound of Theorem~\ref{thm:minaug-hard-greedy}.
As $|\AtomObligs|$ grows the distribution spreads: at $|\AtomObligs|{=}4$,
$79.2\%$ of trials remain optimal and the mean ratio is $1.070$, well below
$H_4\!\approx\!2.08$.
Figure~\ref{fig:minaug}\,(b) shows that allowing cross-atom pair candidates
reduces the mean augmentation size (Greedy $<$ Singleton-only for $|\AtomObligs|\geq 2$),
confirming that richer candidate sets translate into smaller interface augmentations.
Table~\ref{tab:minaug} summarises the per-$|\AtomObligs|$ statistics.

\begin{table}[t]
  \centering
  \caption{\tx{Greedy-MinAug} per-$|\AtomObligs|$ statistics (500 trials each).
    Ratio $=|\GreedySol|/|\OptSol|$; \%-opt $=$ fraction of trials with ratio~$1.0$.
    Mean actions are greedy vs.\ optimal (brute-force).}
  \label{tab:minaug}
  \begin{tabular}{cccccc}
    \toprule
    $|\AtomObligs|$ & \%-opt & Mean ratio & Max ratio & Greedy actions & Opt.\ actions \\
    \midrule
    1 & 100.0 & 1.000 & 1.000 & 1.000 & 1.000 \\
    2 &  98.4 & 1.008 & 1.500 & 1.464 & 1.448 \\
    3 &  89.0 & 1.048 & 1.500 & 2.296 & 2.186 \\
    4 &  79.2 & 1.070 & 1.500 & 3.236 & 3.028 \\
    \bottomrule
  \end{tabular}
\end{table}

\begin{table}[t]
  \centering
  \caption{\tx{Greedy-MinAug} on real datasets (0 actions $=$ certified; attr indices per schema).}
  \label{tab:minaug-realworld}
\begin{tabular}{llccc}
\toprule
Dataset & Query & Certified? & \#Actions & Added attrs \\
\midrule
BibInteg & \texttt{Q\_doi} & yes & 0 & -- \\
 & \texttt{Q\_large\_team} & yes & 0 & -- \\
 & \texttt{Q\_venue} & yes & 0 & -- \\
CrossKG-DBLP & \texttt{Q\_large\_team} & no & 1 & \{2\} \\
 & \texttt{Q\_publisher} & yes & 0 & -- \\
Amazon--Google & \texttt{Q\_catalog} & yes & 0 & -- \\
 & \texttt{Q\_expensive} & no & 1 & \{2\} \\
Fodors--Zagat & \texttt{Q\_cuisine} & no & 1 & \{2\} \\
 & \texttt{Q\_segment} & yes & 0 & -- \\
\bottomrule
\end{tabular}

\end{table}

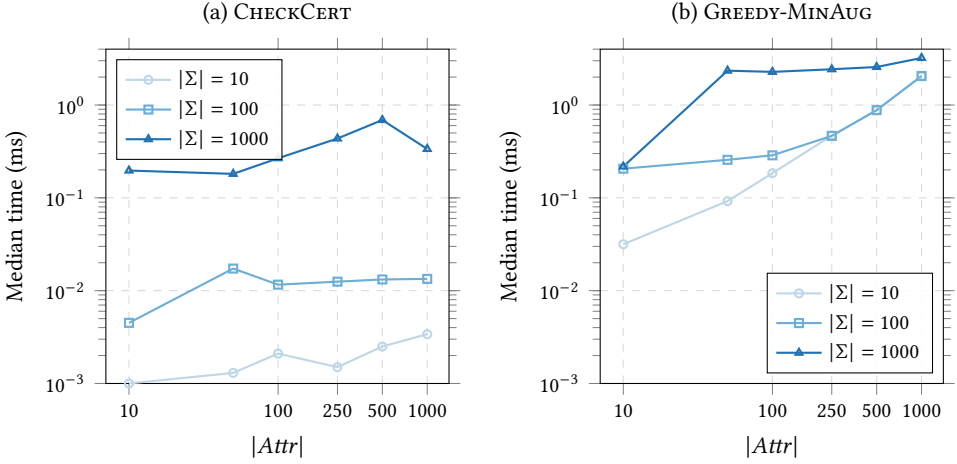
\begin{figure}[t]
  \centering
\definecolor{scBlue0}{RGB}{189,215,231}
\definecolor{scBlue1}{RGB}{107,174,214}
\definecolor{scBlue2}{RGB}{33,113,181}
\begin{tikzpicture}
\begin{groupplot}[
  group style={group size=2 by 1, horizontal sep=2cm},
]
\nextgroupplot[
  width=0.44\textwidth,
  height=6cm,
  xlabel={$|\mathit{Attr}|$},
  ylabel={Median time (ms)},
  xmode=log, ymode=log,
  log basis x=10, log basis y=10,
  xmin=7, xmax=1400,
  ymin=0.001, ymax=4,
  xminorticks=false,
  xticklabel style={font=\footnotesize},
  yticklabel style={font=\footnotesize},
  xlabel style={font=\small},
  ylabel style={font=\small},
  legend style={at={(0.03,0.97)}, anchor=north west, font=\footnotesize,
                legend cell align=left, fill=white, fill opacity=0.9,
                draw opacity=1, text opacity=1},
  grid=major, grid style={dashed, gray!30},
  tick align=outside,
  xtick={10,100,250,500,1000},
  xticklabels={10,100,250,500,1000},
  title={(a) \textsc{CheckCert}},
  title style={font=\small},
]
\addplot[color=scBlue0,mark=o,mark size=1.5pt,line width=0.8pt]
  table[col sep=comma,x=n_attrs,y=cert_f10]{pgfdata/scalability.csv};
\addlegendentry{$|\Sigma|=10$}
\addplot[color=scBlue1,mark=square,mark size=1.5pt,line width=0.8pt]
  table[col sep=comma,x=n_attrs,y=cert_f100]{pgfdata/scalability.csv};
\addlegendentry{$|\Sigma|=100$}
\addplot[color=scBlue2,mark=triangle,mark size=1.5pt,line width=0.8pt]
  table[col sep=comma,x=n_attrs,y=cert_f1000]{pgfdata/scalability.csv};
\addlegendentry{$|\Sigma|=1000$}
\nextgroupplot[
  width=0.44\textwidth,
  height=6cm,
  xlabel={$|\mathit{Attr}|$},
  ylabel={Median time (ms)},
  xmode=log, ymode=log,
  log basis x=10, log basis y=10,
  xmin=7, xmax=1400,
  ymin=0.001, ymax=4,
  xminorticks=false,
  xticklabel style={font=\footnotesize},
  yticklabel style={font=\footnotesize},
  xlabel style={font=\small},
  ylabel style={font=\small},
  legend style={at={(0.97,0.03)}, anchor=south east, font=\footnotesize,
                legend cell align=left, fill=white, fill opacity=0.9,
                draw opacity=1, text opacity=1},
  grid=major, grid style={dashed, gray!30},
  tick align=outside,
  xtick={10,100,250,500,1000},
  xticklabels={10,100,250,500,1000},
  title={(b) \textsc{Greedy-MinAug}},
  title style={font=\small},
]
\addplot[color=scBlue0,mark=o,mark size=1.5pt,line width=0.8pt]
  table[col sep=comma,x=n_attrs,y=minaug_f10]{pgfdata/scalability.csv};
\addlegendentry{$|\Sigma|=10$}
\addplot[color=scBlue1,mark=square,mark size=1.5pt,line width=0.8pt]
  table[col sep=comma,x=n_attrs,y=minaug_f100]{pgfdata/scalability.csv};
\addlegendentry{$|\Sigma|=100$}
\addplot[color=scBlue2,mark=triangle,mark size=1.5pt,line width=0.8pt]
  table[col sep=comma,x=n_attrs,y=minaug_f1000]{pgfdata/scalability.csv};
\addlegendentry{$|\Sigma|=1000$}
\end{groupplot}
\end{tikzpicture}
  \caption{\tx{CheckCert}~(a) and \tx{Greedy-MinAug}~(b) median runtime versus schema size (log--log axes).
  Each line corresponds to a fixed FD count $|\LawSet|\!\in\!\{10,100,1000\}$;
  $x$-axis is the attribute count $|\mathit{Attr}|$.
  Both algorithms stay well below $4\,\mathrm{ms}$ across all $6{\times}6 = 36$ combinations tested.}
  \label{fig:scalability}
\end{figure}

\textbf{Answer to RQ3:} \tx{Greedy-MinAug} achieves near-optimal ratios for all $|\AtomObligs|\!\in\!\{1,2,3,4\}$ (Table~\ref{tab:minaug}; all mean ratios well below the $H_{|\AtomObligs|}$ bound), and richer candidate actions (singletons~+~pairs) reduce augmentation cost over singleton-only proposals. On all four real integration datasets (Table~\ref{tab:minaug-realworld}), certified queries require zero augmentation actions and non-certified queries require exactly 1 action ($<30\,\mu$s each), confirming practical applicability across schema sizes and domains.

\subsection{Scalability (RQ4)}
\label{sec:exp-scalability}

We sweep $|\mathit{Attr}|$ and $|\LawSet|$ each over $\{10,50,100,250,500,1000\}$ ($36$ combinations total).
Certification timing uses random FD schemas; \tx{Greedy-MinAug} timing uses
planted instances with chain FDs that guarantee feasibility (20~candidate actions).
Figure~\ref{fig:scalability} shows three representative $|\LawSet|$ curves.
\tx{CheckCert} peaks at $0.69\,\mathrm{ms}$ median at $(|\mathit{Attr}|,|\LawSet|){=}(500,1000)$,
consistent with its $O(|\mathit{Attr}|{\times}|\LawSet|)$ cost.
\tx{Greedy-MinAug} with 20 candidates peaks at $3.2\,\mathrm{ms}$ at the largest schema
$(1000\!\times\!1000)$, reflecting $O(20\!\times\!|\mathit{Attr}|{\times}|\LawSet|)$ precomputation.
Both algorithms are dominated by the FD-closure fixed-point; at production-scale schemas
with hundreds of attributes and functional dependencies, certification and augmentation
complete in single-digit milliseconds.
\textbf{Answer to RQ4:} Yes---both \tx{CheckCert} and \tx{Greedy-MinAug} remain
practical at database-scale schemas ($|\mathit{Attr}|,|\LawSet|\leq 10^3$) with
sub-millisecond and low-millisecond runtimes respectively.

\subsection{Ablation Study}
\label{sec:exp-ablation}

\paragraph{FD completeness.}
We vary the fraction $\rho$ of true interface laws included in $\LawSet$, sweeping $\rho\in\{0,0.25,0.5,0.75,1.0\}$.
Certificate coverage degrades monotonically as $\rho$ decreases; at $\rho{=}0.5$ it retains the
majority of certifiable queries on the synthetic benchmark.
\tx{CA} error tracks certificate coverage closely, while \tx{MLP} and \tx{ST} degrade more
slowly because they can partially recover missing closure via statistical generalization at large
$m$; the gap closes at small sample sizes.

\subsection{Confirmatory ML Experiments}
\label{sec:exp-ml}

We confirm Theorems~\ref{thm:nonident-minimax} and~\ref{thm:connectivity-emergence} empirically
using the six predictor architectures from Section~\ref{sec:exp-setup}.

\paragraph{Error floor (E1).}
For each $(m,N)$ pair we sample 10 certified and 10 non-certified single-atom Boolean \cq{s},
train each architecture, and report balanced accuracy.
Figure~\ref{fig:error-floor} shows that balanced accuracy remains within $0.03$ of
$\nicefrac{1}{2}$ for \emph{all} architectures on non-certified queries at every $N$, confirming
the irreducible error floor of Theorem~\ref{thm:nonident-minimax}.
\tx{MLP}, \tx{ST}, \tx{GNN-OG}, and \tx{CA} achieve balanced accuracy $1.0$ at
$N{=}5\!\times\!10^4$ on certified queries ($m{=}10$).

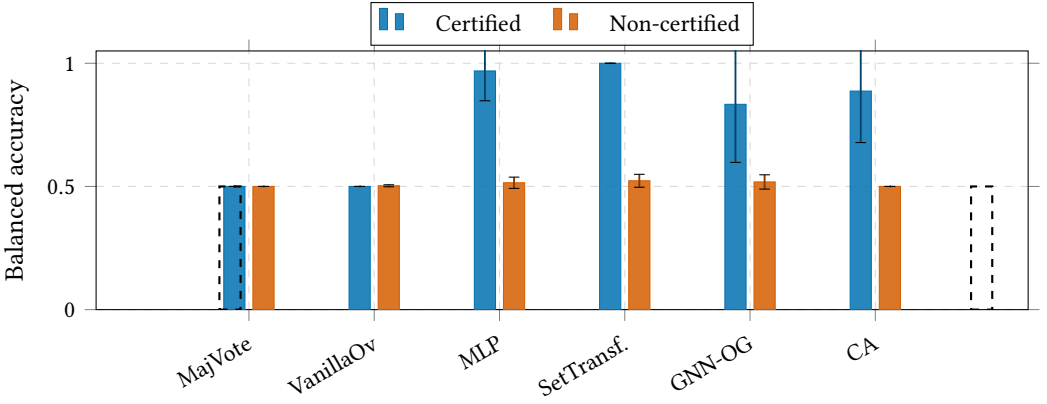
\begin{figure}[t]
  \centering
\definecolor{efCert}{RGB}{0,114,178}
\definecolor{efNonc}{RGB}{213,94,0}
\begin{tikzpicture}
\begin{axis}[
  width=\columnwidth,
  height=5cm,
  ylabel={Balanced accuracy},
  ymin=0, ymax=1.05,
  ybar=3pt,
  bar width=0.28cm,
  enlarge x limits=0.12,
  xtick=data,
  xticklabels={MajVote, VanillaOv, MLP, SetTransf., GNN-OG, CA},
  xticklabel style={rotate=30, anchor=north east, font=\small},
  yticklabel style={font=\small},
  legend style={at={(0.5,1.02)}, anchor=south, font=\small,
                legend cell align=left, fill=white, fill opacity=0.9,
                draw opacity=1, text opacity=1,
                legend columns=-1, column sep=1em},
  grid=major,
  grid style={dashed, gray!30},
  tick align=outside,
]
\addplot[fill=efCert, draw=efCert, fill opacity=0.85,
         error bars/.cd, y dir=both, y explicit,
         error bar style={draw=efCert!60!black, line width=0.7pt}]
  table[x=arch_idx, y=cert_mean, y error=cert_std, col sep=comma]
  {pgfdata/error_floor.csv};
\addlegendentry{Certified}
\addplot[fill=efNonc, draw=efNonc, fill opacity=0.85,
         error bars/.cd, y dir=both, y explicit,
         error bar style={draw=efNonc!60!black, line width=0.7pt}]
  table[x=arch_idx, y=noncert_mean, y error=noncert_std, col sep=comma]
  {pgfdata/error_floor.csv};
\addlegendentry{Non-certified}
\addplot[black, thick, dashed, domain=-0.5:5.5, samples=2, forget plot] {0.5};
\end{axis}
\end{tikzpicture}
  \caption{Balanced accuracy by query type and architecture ($m{=}10$, averaged over
  $N\!\in\!\{10^3,5\!\times\!10^3,5\!\times\!10^4\}$ and 3 seeds).
  Non-certified queries (right bars) are bounded below by $\nicefrac{1}{2}$ for all
  architectures (Theorem~\ref{thm:nonident-minimax}); certified queries (left bars) converge
  to high balanced accuracy for structure-aware architectures.}
  \label{fig:error-floor}
\end{figure}

\paragraph{Capability jumps (E2).}
We use the 5-attribute binary schema with views $V_0{=}\{0,1\}$ and $V_1{=}\{0,2,3,4\}$
overlapping on $\{0\}$, and query $Q{=}\exists x_0.\, R_{V_1}(x_0,0,0,0)$.
FDs $\{0\}\!\to\!k$ are added one at a time for $k\in\{1,2,3,4\}$, producing 5 interface
configurations (steps~0--4); the certificate is satisfied at step~4.
Figure~\ref{fig:capability-jump} shows that the structure-preserving architectures (\tx{MLP}, \tx{ST}, \tx{GNN-OG}, \tx{CA}) remain at ${\approx}\nicefrac{1}{2}$
at steps~0--1 and jump sharply when the certificate is satisfied at step~4, confirming the
structural transition of Theorem~\ref{thm:connectivity-emergence}.
\tx{VanillaOv} differs: it rises above $\nicefrac{1}{2}$ at intermediate steps~2--3 by exploiting statistical correlations introduced by the partial FDs, but collapses back to $\nicefrac{1}{2}$ at step~4, owing to a linear-expressivity failure---$\query$ is an existential query whose answer at step~4 is encoded in tuple-level relational structure that mean-pooled linear features cannot detect.
This separates \emph{structural identifiability} (the closure certificate guarantees the information is present in the interface evidence) from \emph{linear detectability from feature marginals}. The behaviour is query-dependent: on BibInteg (Section~\ref{sec:exp-real}) the certified answers are linearly decodable from the $\LawSet$-closed overlap marginals, and \tx{VanillaOv} reaches $1.0$ there.

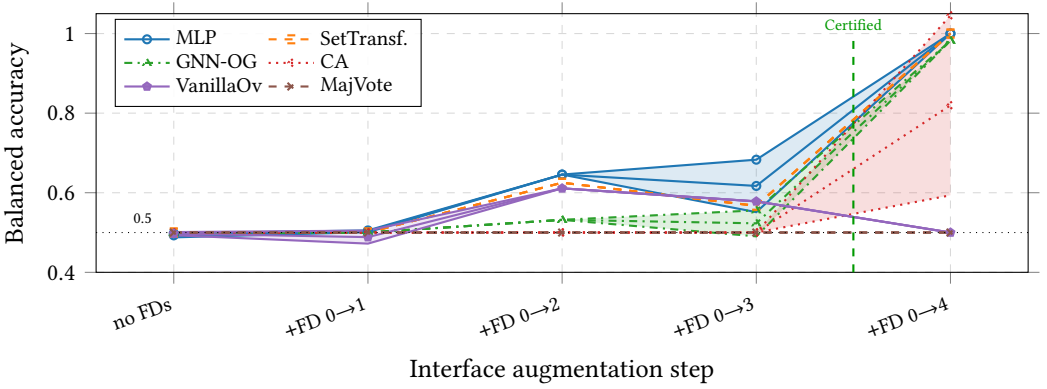
\begin{figure}[t]
  \centering
\definecolor{cjMLP}{RGB}{31,119,180}
\definecolor{cjST}{RGB}{255,127,14}
\definecolor{cjGNN}{RGB}{44,160,44}
\definecolor{cjCA}{RGB}{214,39,40}
\definecolor{cjVoV}{RGB}{148,103,189}
\definecolor{cjMjV}{RGB}{140,86,75}
\begin{tikzpicture}
\begin{axis}[
  width=\columnwidth,
  height=5cm,
  xlabel={Interface augmentation step},
  ylabel={Balanced accuracy},
  xmin=-0.4, xmax=4.4,
  ymin=0.4, ymax=1.05,
  xtick={0,1,2,3,4},
  xticklabels={%
    {no FDs},
    {$+$FD $0{\to}1$},
    {$+$FD $0{\to}2$},
    {$+$FD $0{\to}3$},
    {$+$FD $0{\to}4$}},
  xticklabel style={rotate=20, anchor=north east, font=\footnotesize},
  yticklabel style={font=\small},
  legend style={at={(0.02,0.98)}, anchor=north west, font=\footnotesize,
                legend cell align=left, fill=white, fill opacity=0.9,
                draw opacity=1, text opacity=1, row sep=-2pt},
  legend columns=2,
  grid=major,
  grid style={dashed, gray!30},
  tick align=outside,
]
\addplot[cjMLP, thick, mark=o, mark size=1.5pt, name path=mlpH, forget plot]
  table[x=step, y expr=\thisrow{mlp_mean}+\thisrow{mlp_std}, col sep=comma]
  {pgfdata/capability_jump.csv};
\addplot[cjMLP, thick, name path=mlpL, forget plot]
  table[x=step, y expr=\thisrow{mlp_mean}-\thisrow{mlp_std}, col sep=comma]
  {pgfdata/capability_jump.csv};
\addplot[cjMLP, opacity=0.15, forget plot] fill between[of=mlpH and mlpL];
\addplot[cjMLP, thick, mark=o, mark size=1.5pt]
  table[x=step, y=mlp_mean, col sep=comma]{pgfdata/capability_jump.csv};
\addlegendentry{MLP}
\addplot[cjST, thick, dashed, mark=square, mark size=1.5pt, name path=stH, forget plot]
  table[x=step, y expr=\thisrow{st_mean}+\thisrow{st_std}, col sep=comma]
  {pgfdata/capability_jump.csv};
\addplot[cjST, thick, dashed, name path=stL, forget plot]
  table[x=step, y expr=\thisrow{st_mean}-\thisrow{st_std}, col sep=comma]
  {pgfdata/capability_jump.csv};
\addplot[cjST, opacity=0.15, forget plot] fill between[of=stH and stL];
\addplot[cjST, thick, dashed, mark=square, mark size=1.5pt]
  table[x=step, y=st_mean, col sep=comma]{pgfdata/capability_jump.csv};
\addlegendentry{SetTransf.}
\addplot[cjGNN, thick, dashdotted, mark=triangle, mark size=1.5pt, name path=gnnH, forget plot]
  table[x=step, y expr=\thisrow{gnn_mean}+\thisrow{gnn_std}, col sep=comma]
  {pgfdata/capability_jump.csv};
\addplot[cjGNN, thick, dashdotted, name path=gnnL, forget plot]
  table[x=step, y expr=\thisrow{gnn_mean}-\thisrow{gnn_std}, col sep=comma]
  {pgfdata/capability_jump.csv};
\addplot[cjGNN, opacity=0.15, forget plot] fill between[of=gnnH and gnnL];
\addplot[cjGNN, thick, dashdotted, mark=triangle, mark size=1.5pt]
  table[x=step, y=gnn_mean, col sep=comma]{pgfdata/capability_jump.csv};
\addlegendentry{GNN-OG}
\addplot[cjCA, thick, dotted, mark=diamond, mark size=1.5pt, name path=caH, forget plot]
  table[x=step, y expr=\thisrow{ca_mean}+\thisrow{ca_std}, col sep=comma]
  {pgfdata/capability_jump.csv};
\addplot[cjCA, thick, dotted, name path=caL, forget plot]
  table[x=step, y expr=\thisrow{ca_mean}-\thisrow{ca_std}, col sep=comma]
  {pgfdata/capability_jump.csv};
\addplot[cjCA, opacity=0.15, forget plot] fill between[of=caH and caL];
\addplot[cjCA, thick, dotted, mark=diamond, mark size=1.5pt]
  table[x=step, y=ca_mean, col sep=comma]{pgfdata/capability_jump.csv};
\addlegendentry{CA}
\addplot[cjVoV, thick, name path=vovH, forget plot]
  table[x=step, y expr=\thisrow{vov_mean}+\thisrow{vov_std}, col sep=comma]
  {pgfdata/capability_jump.csv};
\addplot[cjVoV, thick, name path=vovL, forget plot]
  table[x=step, y expr=\thisrow{vov_mean}-\thisrow{vov_std}, col sep=comma]
  {pgfdata/capability_jump.csv};
\addplot[cjVoV, opacity=0.15, forget plot] fill between[of=vovH and vovL];
\addplot[cjVoV, thick, mark=pentagon*, mark size=1.5pt]
  table[x=step, y=vov_mean, col sep=comma]{pgfdata/capability_jump.csv};
\addlegendentry{VanillaOv}
\addplot[cjMjV, thick, dashed, name path=mjvH, forget plot]
  table[x=step, y expr=\thisrow{mjv_mean}+\thisrow{mjv_std}, col sep=comma]
  {pgfdata/capability_jump.csv};
\addplot[cjMjV, thick, dashed, name path=mjvL, forget plot]
  table[x=step, y expr=\thisrow{mjv_mean}-\thisrow{mjv_std}, col sep=comma]
  {pgfdata/capability_jump.csv};
\addplot[cjMjV, opacity=0.15, forget plot] fill between[of=mjvH and mjvL];
\addplot[cjMjV, thick, dashed, mark=x, mark size=2pt]
  table[x=step, y=mjv_mean, col sep=comma]{pgfdata/capability_jump.csv};
\addlegendentry{MajVote}
\addplot[black, thin, dotted, domain=-0.4:4.4, samples=2, forget plot] {0.5}
  node[pos=0.05, above, font=\tiny, black] {0.5};
\draw[green!60!black, thick, dashed] (axis cs:3.5,0.4) -- (axis cs:3.5,1.0);
\node[font=\tiny, green!60!black, anchor=south, fill=white, inner sep=1pt]
  at (axis cs:3.5,1.0) {Certified};
\end{axis}
\end{tikzpicture}
  \caption{Balanced accuracy vs.\ FD-augmentation step.
  Steps~0--1: certificate unsatisfied, all architectures ${\approx}\nicefrac{1}{2}$.
  Step~4: \tx{MLP} and \tx{ST} jump to $1.0$, \tx{GNN-OG} to $0.98$ (identifiability guaranteed, not perfect accuracy), \tx{CA} to $0.82$,
  confirming Theorem~\ref{thm:connectivity-emergence}.}
  \label{fig:capability-jump}
\end{figure}

\section{Conclusion}
\label{sec:conclusion}

We formalized query identifiability for data integration under an explicit relational interface model. A query is identifiable when all legal worlds consistent with the interface evidence return the same answer; when it is not, a minimax lower bound shows that for every non-identifiable query there is a witness pair on which any interface-evidence-only estimator incurs error $\ge\tfrac{1}{2}$.

The closure certificate (Theorem~\ref{thm:frontier-suff}) reduces identifiability checking to forward chaining over functional dependencies, yielding the polynomial-time \tx{CheckCert} algorithm. When a query is not certified, minimum augmentation (Definition~\ref{def:minaug}) asks for the smallest set of new interface actions that would certify it; this reduces to \textsc{Set Cover}, and \tx{Greedy-MinAug} achieves a $H_{|\AtomObligs|}$ approximation guarantee (Theorem~\ref{thm:minaug-hard-greedy}). Experiments confirm that \tx{CheckCert} is exact on exhaustive benchmarks and both algorithms remain practical at schemas with $10^3$ attributes and dependencies.

\paragraph{Scope and limitations.} The closure certificate is both sufficient and necessary on \emph{separable} instances (Corollary~\ref{cor:closure-complete})---instances that actually contain a witness pair if one exists; on non-separable instances a failed certificate is conservative. The CQ footprint model covers single-atom Boolean \cq{s} natively; extending certificates to UCQs, aggregation, or recursive queries requires machinery beyond attribute closure. Interface laws are restricted to FD-syntax (Armstrong implication); denial constraints, inclusion dependencies, and richer tgds are not covered. The real-data validation uses single-tuple and matched-pair worlds ($m{\le}2$), while the confirmatory ML experiments span larger synthetic worlds ($m\in\{10,30,50\}$); real integration datasets with larger world multiplicity remain to be studied. WDC-Product serves as a schema-design case study---its query-critical attributes are synthesized to enforce the certified/uncertified split by construction; BibInteg, CrossKG-DBLP, Amazon-Google, and Fodors-Zagat are the real-data validations.

Several directions remain open. Richer query languages (UCQs, aggregation, recursion) need new certificates beyond footprint closure. Stronger interface augmentation strategies could exploit structured FD sets to reduce outcome multiplicity exponentially. Extending the theory to multi-tuple world grouping and richer integrity constraints remains an important open direction. Identifiability is a design constraint, not only an evaluation metric.

\section*{Reproducibility}
Experiment code is available at \url{https://github.com/danielhz/query-identifiability} and Lean~4 proofs at \url{https://github.com/danielhz/MultiViewIdentifiability}. The formalization machine-checks Theorems~\ref{thm:frontier-suff}, \ref{thm:nonident-minimax}, \ref{thm:rate-lb}, and \ref{thm:robust-threshold}, the general atom-wise closure certificate (\tx{CheckCert} correctness, with Theorem~\ref{thm:frontier-suff} as the single-overlap special case), the semantic content of the interface-visible result (Theorem~\ref{thm:iv-cq-iff}; queries are modeled by their answer-invariant rather than CQ syntax), the FD/Armstrong correspondence (both directions), the reduction of identifiability to query determinacy over overlap-projection views, and Fano's inequality itself (the information inequality behind Corollary~\ref{cor:fano}); three separations are additionally certified as machine-checked disproofs (certificate necessity, union-of-footprints coverage, and MinAug uniqueness). Open in the formalization are the distributional Fano corollary (Corollary~\ref{cor:fano}; its qualitative $\nicefrac{1}{2}$-floor form, Theorem~\ref{thm:nonident-minimax}, is machine-checked), the capability-jump theorem (Theorem~\ref{thm:connectivity-emergence}; only its augmented single-overlap case is checked), and the MinAug greedy $H_k$ bound and NP-hardness (Theorem~\ref{thm:minaug-hard-greedy}). Supplementary material includes BibInteg statistics, extended ML results, and FD-completeness curves; all artifacts (code, Lean proofs, and data) are archived at DaRUS~\cite{darus-artifacts}.

\bibliographystyle{ACM-Reference-Format}
\bibliography{ml2026}


\begin{thebibliography}{37}


\ifx \showCODEN    \undefined \def \showCODEN     #1{\unskip}     \fi
\ifx \showDOI      \undefined \def \showDOI       #1{#1}\fi
\ifx \showISBNx    \undefined \def \showISBNx     #1{\unskip}     \fi
\ifx \showISBNxiii \undefined \def \showISBNxiii  #1{\unskip}     \fi
\ifx \showISSN     \undefined \def \showISSN      #1{\unskip}     \fi
\ifx \showLCCN     \undefined \def \showLCCN      #1{\unskip}     \fi
\ifx \shownote     \undefined \def \shownote      #1{#1}          \fi
\ifx \showarticletitle \undefined \def \showarticletitle #1{#1}   \fi
\ifx \showURL      \undefined \def \showURL       {\relax}        \fi
\providecommand\bibfield[2]{#2}
\providecommand\bibinfo[2]{#2}
\providecommand\natexlab[1]{#1}
\providecommand\showeprint[2][]{arXiv:#2}

\bibitem[\protect\citeauthoryear{Abiteboul and Duschka}{Abiteboul and
  Duschka}{1998}]%
        {abiteboul1998complexity}
\bibfield{author}{\bibinfo{person}{Serge Abiteboul} {and}
  \bibinfo{person}{Oliver~M. Duschka}.} \bibinfo{year}{1998}\natexlab{}.
\newblock \showarticletitle{The Complexity of Answering Queries Using
  Materialized Views}. In \bibinfo{booktitle}{\emph{Proceedings of the
  Seventeenth ACM SIGACT-SIGMOD-SIGART Symposium on Principles of Database
  Systems}}. \bibinfo{pages}{254--263}.
\newblock


\bibitem[\protect\citeauthoryear{Abiteboul, Hull, and Vianu}{Abiteboul
  et~al\mbox{.}}{1995}]%
        {abiteboul1995foundations}
\bibfield{author}{\bibinfo{person}{Serge Abiteboul}, \bibinfo{person}{Richard
  Hull}, {and} \bibinfo{person}{Victor Vianu}.}
  \bibinfo{year}{1995}\natexlab{}.
\newblock \bibinfo{booktitle}{\emph{Foundations of Databases}}.
\newblock \bibinfo{publisher}{Addison-Wesley}.
\newblock


\bibitem[\protect\citeauthoryear{Andrew, Arora, Bilmes, and Livescu}{Andrew
  et~al\mbox{.}}{2013}]%
        {andrew2013deep}
\bibfield{author}{\bibinfo{person}{Galen Andrew}, \bibinfo{person}{Raman
  Arora}, \bibinfo{person}{Jeff Bilmes}, {and} \bibinfo{person}{Karen
  Livescu}.} \bibinfo{year}{2013}\natexlab{}.
\newblock \showarticletitle{Deep Canonical Correlation Analysis}. In
  \bibinfo{booktitle}{\emph{Proceedings of the International Conference on
  Machine Learning (ICML)}}. PMLR, \bibinfo{pages}{1247--1255}.
\newblock


\bibitem[\protect\citeauthoryear{Armstrong}{Armstrong}{1974}]%
        {armstrong1974dependency}
\bibfield{author}{\bibinfo{person}{William~Ward Armstrong}.}
  \bibinfo{year}{1974}\natexlab{}.
\newblock \showarticletitle{Dependency Structures of Data Base Relationships}.
  In \bibinfo{booktitle}{\emph{IFIP Congress}}, Vol.~\bibinfo{volume}{74}.
  Geneva, Switzerland, \bibinfo{pages}{580--583}.
\newblock


\bibitem[\protect\citeauthoryear{Beeri, Fagin, and Howard}{Beeri
  et~al\mbox{.}}{1977}]%
        {beeri1977complete}
\bibfield{author}{\bibinfo{person}{Catriel Beeri}, \bibinfo{person}{Ronald
  Fagin}, {and} \bibinfo{person}{John~H Howard}.}
  \bibinfo{year}{1977}\natexlab{}.
\newblock \showarticletitle{A Complete Axiomatization for Functional and
  Multivalued Dependencies in Database Relations}. In
  \bibinfo{booktitle}{\emph{Proceedings of the 1977 ACM SIGMOD International
  Conference on Management of Data}}. \bibinfo{pages}{47--61}.
\newblock


\bibitem[\protect\citeauthoryear{Blum and Mitchell}{Blum and Mitchell}{1998}]%
        {blum1998combining}
\bibfield{author}{\bibinfo{person}{Avrim Blum} {and} \bibinfo{person}{Tom
  Mitchell}.} \bibinfo{year}{1998}\natexlab{}.
\newblock \showarticletitle{Combining Labeled and Unlabeled Data with
  Co-Training}. In \bibinfo{booktitle}{\emph{Proceedings of the eleventh annual
  conference on Computational learning theory}}. \bibinfo{pages}{92--100}.
\newblock


\bibitem[\protect\citeauthoryear{Brickley, Burgess, and Noy}{Brickley
  et~al\mbox{.}}{2019}]%
        {brickley2019datasetsearch}
\bibfield{author}{\bibinfo{person}{Dan Brickley}, \bibinfo{person}{Matthew
  Burgess}, {and} \bibinfo{person}{Natasha~F. Noy}.}
  \bibinfo{year}{2019}\natexlab{}.
\newblock \showarticletitle{Google Dataset Search: Building a Search Engine for
  Datasets in an Open Web Ecosystem}. In \bibinfo{booktitle}{\emph{The World
  Wide Web Conference ({WWW})}}. \bibinfo{publisher}{{ACM}},
  \bibinfo{pages}{1365--1375}.
\newblock
\urldef\tempurl%
\url{https://doi.org/10.1145/3308558.3313685}
\showDOI{\tempurl}


\bibitem[\protect\citeauthoryear{Chandra and Merlin}{Chandra and
  Merlin}{1977}]%
        {chandra1977optimal}
\bibfield{author}{\bibinfo{person}{Ashok~K Chandra} {and}
  \bibinfo{person}{Philip~M Merlin}.} \bibinfo{year}{1977}\natexlab{}.
\newblock \showarticletitle{Optimal implementation of conjunctive queries in
  relational data bases}. In \bibinfo{booktitle}{\emph{Proceedings of the ninth
  annual ACM symposium on Theory of computing}}. \bibinfo{pages}{77--90}.
\newblock


\bibitem[\protect\citeauthoryear{Christen}{Christen}{2012}]%
        {christen2012data}
\bibfield{author}{\bibinfo{person}{Peter Christen}.}
  \bibinfo{year}{2012}\natexlab{}.
\newblock \showarticletitle{The Data Matching Process}.
\newblock In \bibinfo{booktitle}{\emph{Data matching: concepts and techniques
  for record linkage, entity resolution, and duplicate detection}}.
  \bibinfo{publisher}{Springer}, \bibinfo{pages}{23--35}.
\newblock


\bibitem[\protect\citeauthoryear{Cover and Thomas}{Cover and Thomas}{2006}]%
        {cover2006elements}
\bibfield{author}{\bibinfo{person}{Thomas~M Cover} {and} \bibinfo{person}{Joy~A
  Thomas}.} \bibinfo{year}{2006}\natexlab{}.
\newblock \bibinfo{booktitle}{\emph{Elements of Information Theory}}.
\newblock \bibinfo{publisher}{Wiley-Interscience}.
\newblock


\bibitem[\protect\citeauthoryear{Doan, Halevy, and Ives}{Doan
  et~al\mbox{.}}{2012}]%
        {doan2012principles}
\bibfield{author}{\bibinfo{person}{AnHai Doan}, \bibinfo{person}{Alon Halevy},
  {and} \bibinfo{person}{Zachary Ives}.} \bibinfo{year}{2012}\natexlab{}.
\newblock \bibinfo{booktitle}{\emph{Principles of Data Integration}}.
\newblock \bibinfo{publisher}{Elsevier}.
\newblock


\bibitem[\protect\citeauthoryear{Dong and Naumann}{Dong and Naumann}{2009}]%
        {dong2009data}
\bibfield{author}{\bibinfo{person}{Xin~Luna Dong} {and} \bibinfo{person}{Felix
  Naumann}.} \bibinfo{year}{2009}\natexlab{}.
\newblock \showarticletitle{Data Fusion: Resolving Data Conflicts for
  Integration}.
\newblock \bibinfo{journal}{\emph{Proc. VLDB Endow.}} \bibinfo{volume}{2},
  \bibinfo{number}{2} (\bibinfo{year}{2009}), \bibinfo{pages}{1654--1655}.
\newblock


\bibitem[\protect\citeauthoryear{Fagin, Kolaitis, Miller, and Popa}{Fagin
  et~al\mbox{.}}{2005}]%
        {fagin2005data}
\bibfield{author}{\bibinfo{person}{Ronald Fagin}, \bibinfo{person}{Phokion~G.
  Kolaitis}, \bibinfo{person}{Ren{\'e}e~J. Miller}, {and}
  \bibinfo{person}{Lucian Popa}.} \bibinfo{year}{2005}\natexlab{}.
\newblock \showarticletitle{Data Exchange: Semantics and Query Answering}.
\newblock \bibinfo{journal}{\emph{Theoretical Computer Science}}
  \bibinfo{volume}{336}, \bibinfo{number}{1} (\bibinfo{year}{2005}),
  \bibinfo{pages}{89--124}.
\newblock


\bibitem[\protect\citeauthoryear{Fagin, Ullman, and Vardi}{Fagin
  et~al\mbox{.}}{1983}]%
        {fagin1983semantics}
\bibfield{author}{\bibinfo{person}{Ronald Fagin}, \bibinfo{person}{Jeffrey~D
  Ullman}, {and} \bibinfo{person}{Moshe~Y Vardi}.}
  \bibinfo{year}{1983}\natexlab{}.
\newblock \showarticletitle{On the semantics of updates in databases}. In
  \bibinfo{booktitle}{\emph{Proceedings of the 2nd ACM SIGACT-SIGMOD Symposium
  on Principles of Database Systems}}. \bibinfo{pages}{352--365}.
\newblock


\bibitem[\protect\citeauthoryear{Fano}{Fano}{1961}]%
        {fano1961transmission}
\bibfield{author}{\bibinfo{person}{Robert~M. Fano}.}
  \bibinfo{year}{1961}\natexlab{}.
\newblock \bibinfo{booktitle}{\emph{Transmission of Information: A Statistical
  Theory of Communication}}.
\newblock \bibinfo{publisher}{MIT Press}, \bibinfo{address}{Cambridge, MA}.
\newblock


\bibitem[\protect\citeauthoryear{Ginsburg and Hull}{Ginsburg and Hull}{1983}]%
        {ginsburg1983characterizations}
\bibfield{author}{\bibinfo{person}{Seymour Ginsburg} {and}
  \bibinfo{person}{Richard Hull}.} \bibinfo{year}{1983}\natexlab{}.
\newblock \showarticletitle{Characterizations for functional dependency and
  Boyce-Codd normal form families}.
\newblock \bibinfo{journal}{\emph{Theoretical Computer Science}}
  \bibinfo{volume}{26}, \bibinfo{number}{3} (\bibinfo{year}{1983}),
  \bibinfo{pages}{243--286}.
\newblock


\bibitem[\protect\citeauthoryear{Gogacz and Marcinkowski}{Gogacz and
  Marcinkowski}{2015}]%
        {gogacz2015redspider}
\bibfield{author}{\bibinfo{person}{Tomasz Gogacz} {and} \bibinfo{person}{Jerzy
  Marcinkowski}.} \bibinfo{year}{2015}\natexlab{}.
\newblock \showarticletitle{The Hunt for a Red Spider: Conjunctive Query
  Determinacy Is Undecidable}. In \bibinfo{booktitle}{\emph{30th Annual
  {ACM/IEEE} Symposium on Logic in Computer Science ({LICS})}}.
  \bibinfo{pages}{281--292}.
\newblock
\urldef\tempurl%
\url{https://doi.org/10.1109/LICS.2015.35}
\showDOI{\tempurl}


\bibitem[\protect\citeauthoryear{Gogacz and Marcinkowski}{Gogacz and
  Marcinkowski}{2016}]%
        {gogacz2016rainworm}
\bibfield{author}{\bibinfo{person}{Tomasz Gogacz} {and} \bibinfo{person}{Jerzy
  Marcinkowski}.} \bibinfo{year}{2016}\natexlab{}.
\newblock \showarticletitle{Red Spider Meets a Rainworm: Conjunctive Query
  Finite Determinacy Is Undecidable}. In \bibinfo{booktitle}{\emph{35th {ACM}
  {SIGMOD-SIGACT-SIGAI} Symposium on Principles of Database Systems ({PODS})}}.
  \bibinfo{pages}{121--134}.
\newblock
\urldef\tempurl%
\url{https://doi.org/10.1145/2902251.2902288}
\showDOI{\tempurl}


\bibitem[\protect\citeauthoryear{Halevy, Rajaraman, and Ordille}{Halevy
  et~al\mbox{.}}{2006}]%
        {halevy2006data}
\bibfield{author}{\bibinfo{person}{Alon Halevy}, \bibinfo{person}{Anand
  Rajaraman}, {and} \bibinfo{person}{Joann Ordille}.}
  \bibinfo{year}{2006}\natexlab{}.
\newblock \showarticletitle{Data Integration: The Teenage Years}. In
  \bibinfo{booktitle}{\emph{Proceedings of the 32nd International Conference on
  Very Large Data Bases}}. \bibinfo{pages}{9--16}.
\newblock


\bibitem[\protect\citeauthoryear{Halevy}{Halevy}{2001}]%
        {halevy2001answering}
\bibfield{author}{\bibinfo{person}{Alon~Y. Halevy}.}
  \bibinfo{year}{2001}\natexlab{}.
\newblock \showarticletitle{Answering Queries Using Views: A Survey}.
\newblock \bibinfo{journal}{\emph{The VLDB Journal}} \bibinfo{volume}{10},
  \bibinfo{number}{4} (\bibinfo{year}{2001}), \bibinfo{pages}{270--294}.
\newblock


\bibitem[\protect\citeauthoryear{Le~Cam}{Le~Cam}{2012}]%
        {le2012asymptotic}
\bibfield{author}{\bibinfo{person}{Lucien Le~Cam}.}
  \bibinfo{year}{2012}\natexlab{}.
\newblock \bibinfo{booktitle}{\emph{Asymptotic Methods in Statistical Decision
  Theory}}.
\newblock \bibinfo{publisher}{Springer}.
\newblock


\bibitem[\protect\citeauthoryear{Lee, Lee, Kim, Kosiorek, Choi, and Teh}{Lee
  et~al\mbox{.}}{2019}]%
        {lee2019set}
\bibfield{author}{\bibinfo{person}{Juho Lee}, \bibinfo{person}{Yoonho Lee},
  \bibinfo{person}{Jungtaek Kim}, \bibinfo{person}{Adam~R. Kosiorek},
  \bibinfo{person}{Seungjin Choi}, {and} \bibinfo{person}{Yee~Whye Teh}.}
  \bibinfo{year}{2019}\natexlab{}.
\newblock \showarticletitle{Set {Transformer}: A Framework for Attention-based
  Permutation-Invariant Neural Networks}. In
  \bibinfo{booktitle}{\emph{Proceedings of the 36th International Conference on
  Machine Learning (ICML)}}. \bibinfo{publisher}{PMLR},
  \bibinfo{pages}{3744--3753}.
\newblock


\bibitem[\protect\citeauthoryear{Lenzerini}{Lenzerini}{2002}]%
        {lenzerini2002data}
\bibfield{author}{\bibinfo{person}{Maurizio Lenzerini}.}
  \bibinfo{year}{2002}\natexlab{}.
\newblock \showarticletitle{Data Integration: A Theoretical Perspective}. In
  \bibinfo{booktitle}{\emph{Proceedings of the Twenty-First ACM
  SIGMOD-SIGACT-SIGART Symposium on Principles of Database Systems}}.
  \bibinfo{pages}{233--246}.
\newblock


\bibitem[\protect\citeauthoryear{Levy, Mendelzon, Sagiv, and Srivastava}{Levy
  et~al\mbox{.}}{1995}]%
        {levy1995answering}
\bibfield{author}{\bibinfo{person}{Alon~Y. Levy}, \bibinfo{person}{Alberto~O.
  Mendelzon}, \bibinfo{person}{Yehoshua Sagiv}, {and} \bibinfo{person}{Divesh
  Srivastava}.} \bibinfo{year}{1995}\natexlab{}.
\newblock \showarticletitle{Answering Queries Using Views}. In
  \bibinfo{booktitle}{\emph{Proceedings of the Fourteenth ACM
  SIGACT-SIGMOD-SIGART Symposium on Principles of Database Systems}}.
  \bibinfo{pages}{95--104}.
\newblock


\bibitem[\protect\citeauthoryear{Li, Chang, Ilyas, and Song}{Li
  et~al\mbox{.}}{2001}]%
        {li2001minimizing}
\bibfield{author}{\bibinfo{person}{Chen Li}, \bibinfo{person}{Edward Chang},
  \bibinfo{person}{Ihab~F. Ilyas}, {and} \bibinfo{person}{Jiannan Song}.}
  \bibinfo{year}{2001}\natexlab{}.
\newblock \showarticletitle{Minimizing View Sets without Losing Query-Answering
  Power}. In \bibinfo{booktitle}{\emph{Proceedings of the 8th International
  Conference on Database Theory}}. \bibinfo{pages}{99--113}.
\newblock


\bibitem[\protect\citeauthoryear{Li, Yang, and Zhang}{Li et~al\mbox{.}}{2018}]%
        {li2018survey}
\bibfield{author}{\bibinfo{person}{Yingming Li}, \bibinfo{person}{Ming Yang},
  {and} \bibinfo{person}{Zhongfei Zhang}.} \bibinfo{year}{2018}\natexlab{}.
\newblock \showarticletitle{A Survey of Multi-View Representation Learning}.
\newblock \bibinfo{journal}{\emph{IEEE Transactions on Knowledge and Data
  Engineering}} \bibinfo{volume}{31}, \bibinfo{number}{10}
  (\bibinfo{year}{2018}), \bibinfo{pages}{1863--1883}.
\newblock


\bibitem[\protect\citeauthoryear{Libkin}{Libkin}{2011}]%
        {libkin2011incomplete}
\bibfield{author}{\bibinfo{person}{Leonid Libkin}.}
  \bibinfo{year}{2011}\natexlab{}.
\newblock \showarticletitle{Incomplete Information and Certain Answers in
  General Data Models}. In \bibinfo{booktitle}{\emph{Proceedings of the
  Thirtieth ACM SIGMOD-SIGACT-SIGART Symposium on Principles of Database
  Systems}}. \bibinfo{pages}{59--70}.
\newblock


\bibitem[\protect\citeauthoryear{Lin}{Lin}{1991}]%
        {lin2002divergence}
\bibfield{author}{\bibinfo{person}{Jianhua Lin}.}
  \bibinfo{year}{1991}\natexlab{}.
\newblock \showarticletitle{Divergence Measures Based on the Shannon Entropy}.
\newblock \bibinfo{journal}{\emph{IEEE Transactions on Information Theory}}
  \bibinfo{volume}{37}, \bibinfo{number}{1} (\bibinfo{year}{1991}),
  \bibinfo{pages}{145--151}.
\newblock
\urldef\tempurl%
\url{https://doi.org/10.1109/18.61115}
\showDOI{\tempurl}


\bibitem[\protect\citeauthoryear{Nargesian, Zhu, Miller, Pu, and
  Arocena}{Nargesian et~al\mbox{.}}{2019}]%
        {nargesian2019datalake}
\bibfield{author}{\bibinfo{person}{Fatemeh Nargesian}, \bibinfo{person}{Erkang
  Zhu}, \bibinfo{person}{Ren{\'e}e~J. Miller}, \bibinfo{person}{Ken~Q. Pu},
  {and} \bibinfo{person}{Patricia~C. Arocena}.}
  \bibinfo{year}{2019}\natexlab{}.
\newblock \showarticletitle{Data Lake Management: Challenges and
  Opportunities}.
\newblock \bibinfo{journal}{\emph{Proc. {VLDB} Endow.}} \bibinfo{volume}{12},
  \bibinfo{number}{12} (\bibinfo{year}{2019}), \bibinfo{pages}{1986--1989}.
\newblock
\urldef\tempurl%
\url{https://doi.org/10.14778/3352063.3352116}
\showDOI{\tempurl}


\bibitem[\protect\citeauthoryear{Nash, Segoufin, and Vianu}{Nash
  et~al\mbox{.}}{2010}]%
        {nash2010views}
\bibfield{author}{\bibinfo{person}{Alan Nash}, \bibinfo{person}{Luc Segoufin},
  {and} \bibinfo{person}{Victor Vianu}.} \bibinfo{year}{2010}\natexlab{}.
\newblock \showarticletitle{Views and Queries: Determinacy and Rewriting}.
\newblock \bibinfo{journal}{\emph{ACM Transactions on Database Systems}}
  \bibinfo{volume}{35}, \bibinfo{number}{3} (\bibinfo{year}{2010}),
  \bibinfo{pages}{1--41}.
\newblock


\bibitem[\protect\citeauthoryear{Pasaila}{Pasaila}{2011}]%
        {pasaila2011conjunctive}
\bibfield{author}{\bibinfo{person}{Daniel Pasaila}.}
  \bibinfo{year}{2011}\natexlab{}.
\newblock \showarticletitle{Conjunctive Queries Determinacy and Rewriting}. In
  \bibinfo{booktitle}{\emph{Proceedings of the 14th International Conference on
  Database Theory}}. \bibinfo{pages}{220--231}.
\newblock


\bibitem[\protect\citeauthoryear{Priem, Piwowar, and Orr}{Priem
  et~al\mbox{.}}{2022}]%
        {priem2022openalex}
\bibfield{author}{\bibinfo{person}{Jason Priem}, \bibinfo{person}{Heather
  Piwowar}, {and} \bibinfo{person}{Richard Orr}.}
  \bibinfo{year}{2022}\natexlab{}.
\newblock \bibinfo{title}{{OpenAlex}: A fully-open index of scholarly works,
  authors, venues, institutions, and concepts}.
\newblock \bibinfo{howpublished}{\url{https://openalex.org}}.
\newblock
\newblock
\shownote{arXiv:2205.01833.}


\bibitem[\protect\citeauthoryear{Primpeli and Bizer}{Primpeli and
  Bizer}{2019}]%
        {primpeli2019profiling}
\bibfield{author}{\bibinfo{person}{Anna Primpeli} {and}
  \bibinfo{person}{Christian Bizer}.} \bibinfo{year}{2019}\natexlab{}.
\newblock \showarticletitle{Profiling Entity Matching Benchmark Tasks}. In
  \bibinfo{booktitle}{\emph{Proceedings of the 22nd International Conference on
  Extending Database Technology (EDBT)}}. \bibinfo{pages}{1--12}.
\newblock


\bibitem[\protect\citeauthoryear{Stonebraker and Ilyas}{Stonebraker and
  Ilyas}{2018}]%
        {stonebraker2018integration}
\bibfield{author}{\bibinfo{person}{Michael Stonebraker} {and}
  \bibinfo{person}{Ihab~F. Ilyas}.} \bibinfo{year}{2018}\natexlab{}.
\newblock \showarticletitle{Data Integration: The Current Status and the Way
  Forward}.
\newblock \bibinfo{journal}{\emph{{IEEE} Data Eng. Bull.}}
  \bibinfo{volume}{41}, \bibinfo{number}{2} (\bibinfo{year}{2018}),
  \bibinfo{pages}{3--9}.
\newblock


\bibitem[\protect\citeauthoryear{Thapa and Hern\'{a}ndez}{Thapa and
  Hern\'{a}ndez}{2026}]%
        {darus-artifacts}
\bibfield{author}{\bibinfo{person}{Ratan~Bahadur Thapa} {and}
  \bibinfo{person}{Daniel Hern\'{a}ndez}.} \bibinfo{year}{2026}\natexlab{}.
\newblock \bibinfo{title}{Artifacts for: {Identifiability of Relational Queries
  in Multi-View Pretraining}}.
\newblock \bibinfo{howpublished}{DaRUS, V1}.
\newblock
\urldef\tempurl%
\url{https://doi.org/10.18419/DARUS-6292}
\showDOI{\tempurl}


\bibitem[\protect\citeauthoryear{Tsybakov}{Tsybakov}{2009}]%
        {tsybakov2009introduction}
\bibfield{author}{\bibinfo{person}{A.~B. Tsybakov}.}
  \bibinfo{year}{2009}\natexlab{}.
\newblock \bibinfo{booktitle}{\emph{Introduction to Nonparametric Estimation}}.
\newblock \bibinfo{publisher}{Springer}, \bibinfo{address}{New York}.
\newblock


\bibitem[\protect\citeauthoryear{Zhang, Wang, and Wang}{Zhang
  et~al\mbox{.}}{2025}]%
        {zhang2025augmentation}
\bibfield{author}{\bibinfo{person}{Qi Zhang}, \bibinfo{person}{Yifei Wang},
  {and} \bibinfo{person}{Yisen Wang}.} \bibinfo{year}{2025}\natexlab{}.
\newblock \showarticletitle{An Augmentation Overlap Theory of Contrastive
  Learning}.
\newblock \bibinfo{journal}{\emph{Journal of Machine Learning Research}}
  \bibinfo{volume}{26}, \bibinfo{number}{228} (\bibinfo{year}{2025}),
  \bibinfo{pages}{1--42}.
\newblock


\end{thebibliography}

\clearpage
\appendix

\newcommand{\att}[1]{\attribute{#1}}
\newcommand{\W}{\LegalWorlds}
\newcommand{\Ov}{\OverlapSet}
\newcommand{\clo}[1]{\AttrClosure{#1}}
\newcommand{\aug}[1]{\AugSchema{#1}}
\newcommand{\proj}[2]{#1|_{#2}}
\newcommand{\binent}{H_{\mathrm{b}}}
\newcommand{\nml}{\operatorname{nml}}
\newcommand{\lean}[1]{\texttt{#1}}
\newcommand{\inlean}[1]{\par\smallskip\noindent{\small\textcolor{blue!50!black}{$\vdash$~\textbf{In Lean.}}~#1}\par\smallskip}
\providecommand{\bigland}{\textstyle\bigwedge}
\renewcommand{\Obs}{\operatorname{Obs}}
\renewcommand{\ObsFam}{\Gamma}

\section{Machine-Checked Formalization}\label{a:sec:formal}

This appendix restates the paper's theoretical results with human-readable proofs, each paired with a pointer ($\vdash$~\textbf{In Lean}) to the machine-checked declaration in the \tx{MultiViewIdentifiability} Lean~4 development (\url{https://github.com/danielhz/MultiViewIdentifiability}). To keep the appendix self-contained, the model of Sections~\ref{sec:model}--\ref{sec:identifiability} is recalled briefly before the proofs. Queries are represented \emph{semantically}---by their answer map together with the invariant that the answer depends only on the relevant projections of the world (its footprint)---so ``machine-checked'' means the semantic content of each statement is verified, not that a query syntax is reflected into Lean. Every result below is machine-checked with no \tx{sorry} except those explicitly flagged \emph{open}.

\paragraph{Numbering.} Appendix results are numbered independently of the main text; Table~\ref{a:tab:correspondence} gives the correspondence and Table~\ref{a:tab:deps} the dependency structure. Full Lean declaration names appear at each result's $\vdash$~In~Lean pointer.

\begin{table}[t]
\centering
\caption{Correspondence between the paper's results, this appendix, and the Lean development.}
\label{a:tab:correspondence}
\begin{tabular}{@{}p{0.29\linewidth}llll@{}}
\toprule
Result & Main text & Appx. & Lean module & Status \\
\midrule
FD closure / Armstrong          & \S\ref{sec:model}                  & Thm~\ref{a:thm:armstrong}      & \lean{FDClosure}        & proved \\
Closure certificate             & Thm~\ref{thm:frontier-suff}        & Thm~\ref{a:thm:cert}           & \lean{Certificate}      & proved \\
Certificate not necessary       & Cor~\ref{cor:closure-complete}     & Rem~\ref{a:rem:necessity-false}& \lean{Certificate}      & disproof \\
Atom-wise certificate           & \S\ref{sec:algorithms}             & Thm~\ref{a:thm:atomwise}       & \lean{AtomCertificate}  & proved \\
Union coverage insufficient     & ---                                & Thm~\ref{a:thm:union-false}    & \lean{InterfaceVisible} & disproof \\
Interface-visible fragment      & Thm~\ref{thm:iv-cq-iff}            & Thm~\ref{a:thm:iv}             & \lean{Determinacy}      & proved \\
Determinacy characterisation    & Thm~\ref{thm:iv-cq-iff}            & Thm~\ref{a:thm:det}            & \lean{Determinacy}      & proved \\
Minimax error floor             & Thm~\ref{thm:nonident-minimax}     & Thm~\ref{a:thm:minimax}        & \lean{Minimax}          & proved \\
Outcome / rate bound            & Thm~\ref{thm:rate-lb}              & Thm~\ref{a:thm:outcome}        & \lean{OutcomeBound}     & proved \\
Capacity error bound            & Cor~\ref{cor:fano}                 & Thm~\ref{a:thm:capacity}       & \lean{OutcomeBound}     & proved \\
Fano's inequality               & Cor~\ref{cor:fano}                 & Thm~\ref{a:thm:fano}           & \lean{Entropy}          & proved \\
Distributional Fano corollary   & Cor~\ref{cor:fano}                 & ---                            & ---                     & open \\
Robust threshold                & Thm~\ref{thm:robust-threshold}     & Thm~\ref{a:thm:robust}         & \lean{Information}      & proved \\
Capability jump (augmentation)  & Thm~\ref{thm:connectivity-emergence}& Prop~\ref{a:prop:aug}         & \lean{MinAug}           & aug.\ case \\
MinAug greedy / NP-hardness     & Thm~\ref{thm:minaug-hard-greedy}   & Rem~\ref{a:rem:greedy-open}    & \lean{MinAug}           & open \\
\bottomrule
\end{tabular}
\end{table}

\begin{table}[t]
\centering
\caption{Dependency structure of the appendix's main results.}
\label{a:tab:deps}
\begin{tabular}{@{}p{0.34\linewidth}p{0.58\linewidth}@{}}
\toprule
Result & Depends on \\
\midrule
Closure certificate (Thm~\ref{a:thm:cert})        & Footprint lifting (Lem~\ref{a:lem:footprint-lift}) $\leftarrow$ determinacy under closure (Lem~\ref{a:lem:fd-det}) \\
Atom-wise certificate (Thm~\ref{a:thm:atomwise})  & Atom footprint (Lem~\ref{a:lem:cq-footprint}) \\
Determinacy char.\ (Thm~\ref{a:thm:det})          & Views vs.\ agreement (Lem~\ref{a:lem:view-agree}) \\
Robust threshold (Thm~\ref{a:thm:robust})         & Gibbs (Lem~\ref{a:lem:gibbs}), point-mass (Lem~\ref{a:lem:js-point}), mass concentration (Lem~\ref{a:lem:mass}), unique mode (Lem~\ref{a:lem:mode}) \\
Fano's inequality (Thm~\ref{a:thm:fano})          & Entropy facts (Lem~\ref{a:lem:entropy-basic}), superadditivity (Lem~\ref{a:lem:superadd}) \\
Augmentation (Prop~\ref{a:prop:aug})              & Closure certificate (Thm~\ref{a:thm:cert}) \\
\bottomrule
\end{tabular}
\end{table}

\subsection{The model}\label{a:sec:model}

\paragraph{Attributes, tuples, worlds.}
Fix a universal schema of \emph{attributes} with values in a fixed domain. A \emph{tuple} is a
total map from attributes to values, and a \emph{world} is a set of tuples (a relation over the
universal schema). For an attribute set $X$ and tuples $s,t$, write $s =_X t$ for
\emph{agreement on $X$}: $s(a)=t(a)$ for every $a\in X$. Agreement on $X$ is an equivalence
relation and is antitone in $X$ (agreement on a larger set implies agreement on a smaller one).
\inlean{\lean{Tuple}, \lean{World}, \lean{Tuple.AgreeOn} in \lean{Basic.lean}; the
equivalence and antitonicity facts are \lean{AgreeOn.refl/symm/trans/mono}.}

\paragraph{Interface laws and legality.}
The interface imposes \emph{cross-world} functional dependencies. A dependency
$X\to b$ holds across a class of \emph{legal} worlds $\W$ if, for \emph{any} two legal worlds
$w,w'$ and any tuples $s\in w$, $t\in w'$, agreement on $X$ forces agreement on $b$:
\[
  s =_X t \;\Longrightarrow\; s(b)=t(b).
\]
This is strictly stronger than an instance-level dependency: it ties the value of $b$ to the
value of $X$ \emph{uniformly across all legal worlds}, which is what lets the interface
\emph{determine} attributes rather than merely constrain single instances. A \emph{legality
structure} is a class of legal worlds $\W$ together with a set $\Sigma$ of such laws that all
legal worlds satisfy.
\inlean{\lean{FD}, \lean{FD.HoldsOnPair}, \lean{LegalityStructure} in \lean{Basic.lean}.}

\paragraph{Overlaps, the observable family, observation.}
The interface exposes a family of \emph{observable schemas} $\ObsFam$ --- closure-augmented
designated overlaps (and, more generally, local views or resolver outputs). For an attribute
set $O$ we write $\aug O = \clo{\att O}$ for its closure-augmented schema. Two worlds $w,w'$
\emph{agree on} a schema $O$, written $\proj w O = \proj{w'} O$, when their $O$-projections
coincide as sets of $O$-restricted tuples; the \emph{observation} of $w$ is the family
$\Obs(w)=(\proj w O)_{O\in\ObsFam}$, and $w,w'$ are \emph{observationally equivalent}
($w\sim w'$) when $\Obs(w)=\Obs(w')$, i.e.\ they agree on every schema in $\ObsFam$.
\inlean{\lean{World.AgreeOn} and \lean{ObsEquiv} in \lean{Basic.lean}. In the Lean
development the observable family $\ObsFam$ is the interface's list of designated
(closure-augmented) overlaps \lean{Interface.augOverlaps}; \lean{World.AgreeOn O w w'} is set
equality of the $O$-projections (made precise in \S\ref{a:sec:determinacy}).}

\paragraph{Conjunctive queries, footprint, identifiability.}
A Boolean conjunctive query $Q$ has a \emph{footprint} $\att Q$ --- the attributes its answer
can depend on --- and its answer is invariant under agreement on the footprint: if
$\proj w{\att Q}=\proj{w'}{\att Q}$ then $Q(w)=Q(w')$. $Q$ is \emph{identifiable} under the
interface when observationally equivalent legal worlds always agree on the answer:
\[
  \text{for all legal } w,w':\quad w\sim w' \;\Longrightarrow\; Q(w)=Q(w').
\]
\inlean{\lean{BoolCQ} (with the footprint-faithfulness field) in \lean{Basic.lean};
\lean{Identifiable} in \lean{Identifiability.lean}. As noted above, the footprint-faithfulness
field is the semantic invariant of a \cq, taken as the query's defining property.}

\subsection{Functional-dependency closure and Armstrong entailment}\label{a:sec:fd}

\paragraph{The closure operator.}
For an attribute set $X$, its \emph{closure} $\clo X$ under $\Sigma$ is the least set
containing $X$ and closed under the laws: if $Y\to b\in\Sigma$ and $Y\subseteq\clo X$ then
$b\in\clo X$. The closure is \emph{extensive} ($X\subseteq\clo X$), \emph{monotone}
($X\subseteq Y\Rightarrow\clo X\subseteq\clo Y$), and \emph{idempotent}
($\clo{\clo X}=\clo X$).
\inlean{\lean{InClosure}/\lean{fdClosure} with \lean{fdClosure\_extensive},
\lean{fdClosure\_mono}, \lean{fdClosure\_idem\_le}, \lean{fdClosure\_idem\_ge} in
\lean{FDClosure.lean}.}

\begin{lemma}[Determinacy under closure]\label{a:lem:fd-det}
For every attribute set $X$, all legal worlds $w,w'$, and all tuples $s\in w$, $t\in w'$,
\[
  s =_X t \;\Longrightarrow\; s =_{\clo X} t.
\]
\end{lemma}
\begin{proof}
Induct on the construction of $\clo X$. If $b\in X$ the conclusion is the hypothesis. If
$b$ enters the closure through a law $Y\to b$ with $Y\subseteq\clo X$, then by the induction
hypothesis $s =_Y t$; since the law holds across the legal worlds $w,w'$ and $s\in w$,
$t\in w'$, agreement on $Y$ forces $s(b)=t(b)$. Hence $s =_{\clo X} t$.
\end{proof}
\inlean{\lean{fd\_determinacy} in \lean{FDClosure.lean} (with corollaries
\lean{fdClosure\_propagates\_agreement} lifting it from a single attribute to the whole
closure).}

\begin{theorem}[Armstrong soundness and completeness]\label{a:thm:armstrong}
For attribute sets $X$ and a single attribute $a$, the law $X\to a$ is entailed by $\Sigma$
(holds in every legality structure satisfying $\Sigma$) \emph{iff} $a\in\clo X$.
\end{theorem}
\begin{proof}
\emph{Soundness} ($a\in\clo X\Rightarrow$ entailed): immediate from
Lemma~\ref{a:lem:fd-det}, which already shows $s=_X t \Rightarrow s(a)=t(a)$ for any legality
structure satisfying $\Sigma$.

\emph{Completeness} ($\text{entailed}\Rightarrow a\in\clo X$): contrapositive. If
$a\notin\clo X$, build a \emph{canonical} two-tuple legality structure that satisfies $\Sigma$
but violates $X\to a$: take two tuples that agree exactly on $\clo X$ and differ on every
attribute outside it. This pair satisfies every law of $\Sigma$ (a law $Y\to b$ with the
tuples agreeing on $Y$ forces $Y\subseteq\clo X$, hence $b\in\clo X$, hence agreement on $b$),
yet they agree on $X$ and disagree on $a\notin\clo X$. So $X\to a$ is not entailed.
\end{proof}
\inlean{\lean{fdClosure\_sound} and \lean{fdClosure\_complete} in \lean{FDClosure.lean}; the
canonical separating structure is \lean{canonStruct}.}

\subsection{The closure certificate}\label{a:sec:cert}

The first sufficient condition for identifiability covers a query whose entire footprint is
captured by the closure of a single observable overlap.

\begin{lemma}[Footprint lifting]\label{a:lem:footprint-lift}
Let $O$ be an attribute set and $Q$ a query with $\att Q\subseteq\clo O$. For all legal
worlds $w,w'$ and tuples $s\in w$, $t\in w'$, if $s =_O t$ then $s =_{\att Q} t$.
\end{lemma}
\begin{proof}
By Lemma~\ref{a:lem:fd-det}, $s =_O t$ gives $s =_{\clo O} t$; antitonicity of agreement and
$\att Q\subseteq\clo O$ give $s =_{\att Q} t$.
\end{proof}
\inlean{\lean{footprint\_lifting} in \lean{Certificate.lean}.}

\begin{theorem}[Closure certificate]\label{a:thm:cert}
If $\att Q\subseteq\clo O$ for some observable overlap $O\in\ObsFam$, then $Q$ is identifiable.
\end{theorem}
\begin{proof}
Let $w,w'$ be legal with $w\sim w'$. Since $O\in\ObsFam$, observation equivalence gives
$\proj w O=\proj{w'} O$: every tuple of $w$ has a tuple of $w'$ agreeing on $O$, and
vice versa. Take $s\in w$; pick the matching $t\in w'$ with $s =_O t$. By
Lemma~\ref{a:lem:footprint-lift}, $s =_{\att Q} t$. Symmetrically every tuple of $w'$ has a
footprint-agreeing partner in $w$. Hence $\proj w{\att Q}=\proj{w'}{\att Q}$, and footprint
faithfulness yields $Q(w)=Q(w')$.
\end{proof}
\inlean{\lean{certificate\_sufficiency} in \lean{Certificate.lean}. The ``fully grounded''
variant, where every part of the query is grounded in some overlap, is
\lean{fully\_grounded\_identifiable}/\lean{hasCertificate\_identifiable}.}

\begin{remark}[The certificate is sufficient, not necessary]\label{a:rem:necessity-false}
The converse fails: an identifiable query need not have its footprint inside any single
closure. The Lean development exhibits a machine-checked counterexample --- a degenerate
single-world interface in which every query is vacuously identifiable while the closure
condition fails. Necessity holds only under an additional richness (``closure-separability'')
assumption on the legal class, which is not assumed here.
\inlean{\lean{certificate\_necessity\_false} in \lean{Certificate.lean}.}
\end{remark}

\subsection{The atom-wise closure certificate}\label{a:sec:atom}

The closure certificate of \S\ref{a:sec:cert} asks the \emph{whole} footprint to sit inside one
overlap. The atom-wise certificate is sharper: it covers a query \emph{atom by atom}, allowing
different atoms to be grounded in different overlaps. This is the right granularity, because
covering a query's attributes \emph{separately} is not enough.

A conjunctive query is presented by its relation-atom schemas $U_1,\dots,U_m$ together with an
answer that depends only on the per-atom projections.

\begin{lemma}[Atom footprint]\label{a:lem:cq-footprint}
Let $Q$ have atom schemas $U_1,\dots,U_m$. If legal worlds $w,w'$ satisfy
$\proj w{U_j}=\proj{w'}{U_j}$ for every $j$, then $Q(w)=Q(w')$.
\end{lemma}
\begin{proof}
The interpretation of every relation symbol occurring in $Q$ is fixed by the per-atom
projections; agreeing on all of them makes the two induced structures interpret $Q$
identically, so set-semantics evaluation gives the same answer. (In the formalisation this is
the query's defining per-atom faithfulness invariant.)
\end{proof}
\inlean{\lean{cq\_footprint} (the field \lean{AtomCQ.faithful}) in \lean{AtomCertificate.lean}.}

\begin{theorem}[Atom-wise closure certificate]\label{a:thm:atomwise}
If for every atom $U_j$ there is an observable overlap $O_j\in\ObsFam$ with
$\att{U_j}\subseteq\clo{O_j}$, then $Q$ is identifiable.
\end{theorem}
\begin{proof}
Let $w,w'$ be legal with $w\sim w'$. Fix an atom $U_j$ and its overlap $O_j$. Observation
equivalence gives $\proj w{O_j}=\proj{w'}{O_j}$, and the lifting argument of
Lemma~\ref{a:lem:footprint-lift} (now with $\att{U_j}\subseteq\clo{O_j}$) upgrades this to
$\proj w{U_j}=\proj{w'}{U_j}$: each tuple of $w$ has an $O_j$-agreeing partner in $w'$, which
by Lemma~\ref{a:lem:fd-det} agrees on $\clo{O_j}\supseteq\att{U_j}$, and conversely. As this
holds for every atom, Lemma~\ref{a:lem:cq-footprint} yields $Q(w)=Q(w')$.
\end{proof}
\inlean{\lean{atomwise\_certificate} in \lean{AtomCertificate.lean}; the per-overlap lifting
step is \lean{agreeOn\_lift}. The single-overlap certificate (Theorem~\ref{a:thm:cert}) is the
one-atom case. This declaration depends on no axioms.}

\begin{theorem}[Union-of-footprint coverage is insufficient]\label{a:thm:union-false}
It is \emph{not} the case that ``$\att Q\subseteq\clo{\bigcup_{O\in\ObsFam}\att O}$ implies
$Q$ identifiable.'' That is, covering the footprint by the closure of the \emph{union} of the
overlaps does not suffice.
\end{theorem}
\begin{proof}
Counterexample. Take two overlaps $\{0\}$ and $\{1\}$, no laws, and the two worlds
\[
  w_A=\{(0,0),(1,1)\},\qquad w_B=\{(0,1),(1,0)\}
\]
(pairs written as $(\text{attr}_0,\text{attr}_1)$). Their projections onto $\{0\}$ are both
$\{0,1\}$, and likewise onto $\{1\}$, so $w_A\sim w_B$. The query ``some tuple has
$\text{attr}_0=\text{attr}_1$'' is true on $w_A$ (e.g.\ $(0,0)$) and false on $w_B$, although
its footprint $\{0,1\}$ lies in the closure of the union $\{0\}\cup\{1\}$. Hence the union
condition does not imply identifiability.
\end{proof}
\inlean{\lean{union\_footprint\_coverage\_insufficient} in \lean{InterfaceVisible.lean}.
This refutes the flat union-of-footprint heuristic only; it is consistent with both the
atom-wise certificate (the offending atom $\{0,1\}$ lies in \emph{no} single overlap) and the
interface-visible fragment below.}

\subsection{The interface-visible fragment}\label{a:sec:iv}

A complementary, exact fragment is obtained by building queries directly from the observation.
For each observable schema $H\in\ObsFam$ introduce a predicate $R_H$ interpreted on $w$ as the
projection $\proj w H$; the \emph{interface-visible vocabulary} $\mathcal L_{\mathrm{iv}}$
consists of these base predicates together with derived predicates, each defined by a
conjunctive query over the $R_H$. A conjunctive query is \emph{interface-visible} when every
relation symbol it uses belongs to $\mathcal L_{\mathrm{iv}}$.

\begin{theorem}[Interface-visible fragment is identifiable]\label{a:thm:iv}
Every interface-visible conjunctive query is identifiable.
\end{theorem}
\begin{proof}
Let $w\sim w'$ be legal. For each base symbol $R_H$, observation equivalence gives
$\proj w H=\proj{w'} H$, so the two structures interpret $R_H$ identically. Each derived
symbol is computed by a fixed query over the base symbols, hence is also interpreted
identically. Thus the two structures agree on every symbol the query mentions, and
set-semantics evaluation gives $Q(w)=Q(w')$.
\end{proof}
\inlean{\lean{iv\_identifiable} (with the query model \lean{IVQuery}, whose \lean{visible}
field records that the answer is fixed by the observable projections) in
\lean{Determinacy.lean}. Unlike the closure certificate this needs no legality assumption ---
the symbols are fixed by the observation outright --- and it permits joins \emph{across}
overlaps, provided they are expressed over the observable $R_H$ relations rather than over raw
attributes spanning overlaps (the case Theorem~\ref{a:thm:union-false} rules out).}

\subsection{Identifiability is query determinacy}\label{a:sec:determinacy}

Identifiability coincides exactly with \emph{query determinacy} (in the sense of
Nash--Segoufin--Vianu) by the overlap-projection views. Define the $O$-projection view of a
world by $\mathrm{view}_O(w)=\{\,t : \exists s\in w,\ s=_O t\,\}$.

\begin{lemma}[Views vs.\ agreement]\label{a:lem:view-agree}
$\mathrm{view}_O(w)=\mathrm{view}_O(w')$ if and only if $\proj w O=\proj{w'} O$ (the two
worlds agree on $O$).
\end{lemma}
\begin{proof}
($\Leftarrow$) If the worlds agree on $O$ and $t\in\mathrm{view}_O(w)$ via $s\in w$ with
$s=_O t$, then $s$ has an $O$-agreeing partner $u\in w'$, and $u=_O t$ by transitivity, so
$t\in\mathrm{view}_O(w')$; symmetrically. ($\Rightarrow$) For $s\in w$ we have
$s\in\mathrm{view}_O(w)$ (reflexivity), hence $s\in\mathrm{view}_O(w')$, giving a partner
$u\in w'$ with $u=_O s$; symmetrically. Thus the worlds agree on $O$.
\end{proof}
\inlean{\lean{projView} and \lean{projView\_eq\_iff\_agreeOn} in \lean{Determinacy.lean}.}

\begin{theorem}[Determinacy characterisation]\label{a:thm:det}
A query is identifiable under the interface if and only if it is determined, in the
determinacy sense, by the overlap-projection views: every two legal worlds with equal views
agree on the answer.
\end{theorem}
\begin{proof}
By Lemma~\ref{a:lem:view-agree}, equality of all overlap-projection views is the same relation
as observation equivalence. Identifiability is precisely ``observationally equivalent legal
worlds agree on the answer,'' which is exactly determinacy by those views.
\end{proof}
\inlean{\lean{DeterminedBy} and \lean{identifiable\_iff\_determined} in
\lean{Determinacy.lean}.}

\subsection{The minimax error floor}\label{a:sec:minimax}

Identifiability is also exactly the line below which no observation-based predictor can do
better than chance.

\begin{theorem}[Minimax floor]\label{a:thm:minimax}
Let $Q$ be \emph{non}-identifiable, witnessed by legal $w\sim w'$ with $Q(w)\neq Q(w')$, and
let $f$ be any classifier that depends only on the observation (so $f(w)=f(w')$ whenever
$w\sim w'$). Then $f$ misclassifies at least one of $w,w'$: there is a legal world on which
$f$'s prediction is wrong.
\end{theorem}
\begin{proof}
Because $f$ is observation-based and $w\sim w'$, we have $f(w)=f(w')$. But $Q(w)\neq Q(w')$,
so $f$'s single shared prediction cannot match both answers; it is wrong on at least one of
the two legal worlds. Averaged over the witness pair this is an error rate of at least
$\tfrac12$.
\end{proof}
\inlean{\lean{minimax\_error\_floor} in \lean{Minimax.lean} (with
\lean{not\_perfect\_balanced\_accuracy}); ``observation-based'' is \lean{ObsDetermined}.}

\subsection{Robustness: a zero-discrepancy threshold}\label{a:sec:robust}

The certificates are exact: zero overlap loss forces exact agreement on certified queries.
We show a quantitative refinement --- a \emph{small} loss already forces exact agreement ---
when prediction quality is measured by Jensen--Shannon divergence. All divergences use natural
logarithms ($\kappa=1$).

For finite mass functions $p,q$, the Kullback--Leibler and Jensen--Shannon divergences are
\[
  \KL(p\,\|\,q)=\sum_a p(a)\log\frac{p(a)}{q(a)},
  \qquad
  \JS(p\,\|\,q)=\tfrac12\KL(p\,\|\,m)+\tfrac12\KL(q\,\|\,m),\quad m=\tfrac{p+q}{2}.
\]
Write $\delta_x$ for the point mass at $x$.

\begin{lemma}[Gibbs' inequality]\label{a:lem:gibbs}
For mass functions $p,q$ with $q(a)>0$ whenever $p(a)\neq 0$, $\KL(p\,\|\,q)\ge 0$.
\end{lemma}
\begin{proof}
Using $\log z\le z-1$ for $z>0$, for each $a$ with $p(a)>0$,
$p(a)\log\frac{q(a)}{p(a)}\le p(a)\big(\frac{q(a)}{p(a)}-1\big)=q(a)-p(a)$; the inequality also
holds trivially when $p(a)=0$ (left side $0\le q(a)$). Summing,
$\sum_a p(a)\log\frac{q(a)}{p(a)}\le\sum_a(q(a)-p(a))=1-1=0$, and $\KL(p\,\|\,q)$ is the
negation of the left-hand side, hence $\ge 0$.
\end{proof}
\inlean{\lean{kl\_nonneg} in \lean{Information.lean}.}

\begin{lemma}[Point-mass lower bound]\label{a:lem:js-point}
For a mass function $p$ and any $x$, $\ \JS(\delta_x\,\|\,p)\ \ge\ \tfrac12\log\!\frac{2}{1+p(x)}$.
\end{lemma}
\begin{proof}
With $m=\tfrac12(\delta_x+p)$ we have $m(x)=\tfrac{1+p(x)}2$, and
$\KL(\delta_x\,\|\,m)=\log\frac1{m(x)}=\log\frac{2}{1+p(x)}$. By Lemma~\ref{a:lem:gibbs},
$\KL(p\,\|\,m)\ge0$, so
$\JS(\delta_x\,\|\,p)\ge\tfrac12\KL(\delta_x\,\|\,m)=\tfrac12\log\frac2{1+p(x)}$.
\end{proof}
\inlean{\lean{kl\_dirac\_mix} and \lean{jsdiv\_dirac\_lower} in \lean{Information.lean}.}

\begin{lemma}[Mass concentration]\label{a:lem:mass}
If $\JS(\delta_x\,\|\,p)\le\gamma$ with $\gamma<\tfrac18$, then $p(x)>\tfrac12$.
\end{lemma}
\begin{proof}
By Lemma~\ref{a:lem:js-point}, $\tfrac12\log\frac2{1+p(x)}\le\gamma$, so
$\frac2{1+p(x)}\le e^{2\gamma}$ and $1+p(x)\ge 2e^{-2\gamma}$. Using $e^{-t}\ge 1-t$,
$1+p(x)\ge 2(1-2\gamma)$, i.e.\ $p(x)\ge 1-4\gamma$. As $\gamma<\tfrac18$, $p(x)>\tfrac12$.
(Avoiding the usual $\sqrt{\cdot}$ Pinsker step keeps the bound elementary.)
\end{proof}
\inlean{\lean{px\_gt\_half} in \lean{Information.lean}.}

\begin{lemma}[Unique majority / unique mode]\label{a:lem:mode}
A finite mass function has at most one outcome of mass $>\tfrac12$. Consequently, if
$\JS(\delta_x\,\|\,p)\le\gamma$ and $\JS(\delta_{x'}\,\|\,p)\le\gamma$ with $\gamma<\tfrac18$,
then $x=x'$.
\end{lemma}
\begin{proof}
If $x\neq x'$ both had mass $>\tfrac12$, then $p(x)+p(x')>1$, contradicting $\sum_a p(a)=1$.
The second statement applies Lemma~\ref{a:lem:mass} to $x$ and $x'$.
\end{proof}
\inlean{\lean{unique\_majority} and \lean{js\_mode} in \lean{Information.lean}.}

\begin{theorem}[Robust threshold]\label{a:thm:robust}
Suppose the overlap projection $\mathrm{proj}(w)$ of a world (its value on a closure-augmented
overlap $\aug O$) covers the query's footprint, in the sense that
$\mathrm{proj}(w)=\mathrm{proj}(w')\Rightarrow Q(w)=Q(w')$, and that the overlap loss is
anchored to a fixed reference $p_O$: $\ \eta\cdot\JS(\delta_{\mathrm{proj}(w)}\,\|\,p_O)\le
\ell(w)$ with $\eta>0$. Let $\varepsilon_0=\eta/8$. Then for every $\varepsilon<\varepsilon_0$,
$Q$ is $(\varepsilon,0)$-identifiable: any two worlds with loss $\le\varepsilon$ give the same
answer.
\end{theorem}
\begin{proof}
If $\ell(w),\ell(w')\le\varepsilon<\eta/8$, anchoring gives
$\JS(\delta_{\mathrm{proj}(w)}\,\|\,p_O)\le\varepsilon/\eta<\tfrac18$ and likewise for $w'$. By
Lemma~\ref{a:lem:mode}, $\mathrm{proj}(w)=\mathrm{proj}(w')$, and footprint coverage yields
$Q(w)=Q(w')$.
\end{proof}
\inlean{\lean{robust\_threshold} in \lean{Information.lean}. The footprint-coverage hypothesis
is the abstract counterpart of ``$\att Q\subseteq\aug O$'' (Theorem~\ref{a:thm:cert}); the
anchoring hypothesis models the overlap-anchored loss term.}

\subsection{Information-theoretic lower bounds}\label{a:sec:info}

\paragraph{Capacity and outcome bounds.}
Let $m_Q$ be the \emph{outcome multiplicity} of $Q$ --- the number of distinct answers it
realises over the legal worlds.

\begin{theorem}[Outcome lower bound]\label{a:thm:outcome}
Any predictor that reads the observation, stores it in a representation with at most $2^k$
states, and answers $Q$ correctly on every legal world satisfies $m_Q\le 2^k$ (equivalently
$k\ge\log_2 m_Q$).
\end{theorem}
\begin{proof}
Write the predictor as a representation map $\mathrm{rep}$ into a set $R$ with $|R|\le 2^k$,
followed by a decoder $\mathrm{ans}:R\to\text{answers}$, correct in the sense
$\mathrm{ans}(\mathrm{rep}(w))=Q(w)$ for every legal $w$. Then the realised answers are
$\{Q(w)\}=\{\mathrm{ans}(\mathrm{rep}(w))\}\subseteq\mathrm{range}(\mathrm{ans})$, a set of size
$\le|R|\le 2^k$. Hence $m_Q\le 2^k$.
\end{proof}
\inlean{\lean{outcome\_lower\_bound} in \lean{OutcomeBound.lean}.}

\begin{theorem}[Capacity error bound]\label{a:thm:capacity}
Index the $m_Q$ distinct outcomes by a set on which the true answer is injective. Any
$\le 2^k$-state predictor errs on at least $m_Q-2^k$ of them; under a uniform prior its error
satisfies $P_e\ge 1-2^k/m_Q$.
\end{theorem}
\begin{proof}
On the set of correctly-answered outcomes the representation map is injective (two correctly
answered, distinct-outcome worlds with the same representation would force equal answers, hence
equal outcomes), so at most $|R|\le 2^k$ are correct; the remaining $\ge m_Q-2^k$ are errors.
Dividing by $m_Q$ gives $P_e\ge 1-2^k/m_Q$.
\end{proof}
\inlean{\lean{capacity\_error\_bound} and \lean{capacity\_error\_rate} in \lean{OutcomeBound.lean}.}

\paragraph{Fano's inequality.}
Write $\nml(x)=-x\log x$ and entropy $H(p)=\sum_a\nml(p(a))$ (nats); the binary entropy is
$\binent(p)=\nml(p)+\nml(1-p)$.

\begin{lemma}[Basic entropy facts]\label{a:lem:entropy-basic}
For a mass function on an $M$-element set: $H\ge 0$; $H(\text{uniform})=\log M$; and
$H(p)\le\log M$ (maximum entropy).
\end{lemma}
\begin{proof}
$\nml\ge0$ on $[0,1]$, so $H\ge0$. For the uniform law each term is $\nml(1/M)=\frac1M\log M$,
summing to $\log M$. Maximum entropy is $\KL(p\,\|\,\text{uniform})\ge0$
(Lemma~\ref{a:lem:gibbs}): expanding, $\KL(p\,\|\,\text{uniform})=\log M-H(p)$.
\end{proof}
\inlean{\lean{entropy\_nonneg}, \lean{entropy\_uniform}, \lean{entropy\_le\_log\_card} in
\lean{Entropy.lean}.}

\begin{lemma}[Superadditivity and the max-entropy-with-total bound]\label{a:lem:superadd}
$\nml$ is superadditive on nonnegatives: $\nml(x+y)\le\nml(x)+\nml(y)$. Hence joint entropy
dominates a marginal, $H(\text{marginal})\le H(\text{joint})$; and for nonnegatives
$a_1,\dots,a_N$ with sum $S$, $\ \sum_i\nml(a_i)\le\nml(S)+S\log N$.
\end{lemma}
\begin{proof}
For $x,y>0$, since $0<x,y\le x+y$ and $\log$ is increasing,
$x\log x+y\log y\le x\log(x+y)+y\log(x+y)=(x+y)\log(x+y)$, which is
$\nml(x+y)\le\nml(x)+\nml(y)$ (boundary cases $x=0$ or $y=0$ are immediate). Summing
row-wise gives $H(\text{marginal})\le H(\text{joint})$. The last bound is Jensen for the
concave $\nml$ with uniform weights: $\frac1N\sum_i\nml(a_i)\le\nml(\frac SN)$, and
$N\,\nml(\frac SN)=\nml(S)+S\log N$.
\end{proof}
\inlean{\lean{negMulLog\_add\_le}, \lean{entropy\_marginalX\_le}/\lean{condEntropy\_nonneg},
\lean{negMulLog\_sum\_le\_total} in \lean{Entropy.lean}.}

\begin{theorem}[Fano's inequality]\label{a:thm:fano}
Let $q$ be a mass function on an $M$-element set with a distinguished outcome $x_0$, and write
$p=1-q(x_0)$ for the residual (``error'') mass. Then
$\ H(q)\le\binent(p)+p\log(M-1)$.
\end{theorem}
\begin{proof}
Split off $x_0$: $H(q)=\nml(q(x_0))+\sum_{a\neq x_0}\nml(q(a))$. The residual mass is
$\sum_{a\neq x_0}q(a)=p$ over $M-1$ outcomes, so by the max-entropy-with-total bound
(Lemma~\ref{a:lem:superadd}), $\sum_{a\neq x_0}\nml(q(a))\le\nml(p)+p\log(M-1)$. Therefore
$H(q)\le\nml(q(x_0))+\nml(p)+p\log(M-1)=\nml(1-p)+\nml(p)+p\log(M-1)
=\binent(p)+p\log(M-1)$.
\end{proof}
\inlean{\lean{entropy\_le\_fano} (with \lean{binEntropy\_eq\_negMulLog}) in \lean{Entropy.lean}.}

\begin{remark}[Distributional predictor corollary --- open]\label{a:rem:fano-open}
The verbatim entropic predictor bound $P_e\ge 1-\frac{I(Q;\mathrm{Obs})+1}{\log_2 m_Q}$, the
conditional/averaged lift of Theorem~\ref{a:thm:fano} over a joint distribution, is \emph{not}
formalised. Its operational content (interface capacity forces an error floor) is the proved
Theorem~\ref{a:thm:capacity}.
\end{remark}

\subsection{Minimum augmentation}\label{a:sec:minaug}

When a query is not certified, one seeks the smallest \emph{augmentation} --- extra attributes
added to an overlap --- whose closure covers the footprint. Say $\mathit{aug}$ is an
\emph{augmentation certificate} for $Q$ on overlap $O$ when $\att Q\subseteq\clo{O\cup
\mathit{aug}}$.

\begin{proposition}[Structure of augmentation certificates]\label{a:prop:aug}
Augmentation certificates are monotone (if $\mathit{aug}_1\subseteq\mathit{aug}_2$ certifies,
so does $\mathit{aug}_2$); certification depends only on the closure $\clo{O\cup\mathit{aug}}$;
and $\mathit{aug}$ certifies $Q$ iff it covers the \emph{residual}
$\{a\in\att Q : a\notin\clo O\}$. Moreover, adding $\mathit{aug}$ to the interface (making
$O\cup\mathit{aug}$ an observable overlap) renders $Q$ identifiable \emph{under the augmented
interface}.
\end{proposition}
\begin{proof}
Monotonicity and closure-dependence are immediate from monotonicity and idempotence of $\clo
\cdot$. The residual characterisation splits $\att Q$ into the part already in $\clo O$ and the
rest. Identifiability under the augmented interface is Theorem~\ref{a:thm:cert} applied to the
overlap $O\cup\mathit{aug}$ in the augmented interface.
\end{proof}
\inlean{\lean{augCertificate\_mono}, \lean{augCertificate\_closure\_char},
\lean{augCertificate\_iff\_covers\_residual}, \lean{aug\_closure\_equiv},
\lean{augCertificate\_identifiable\_augmented} in \lean{MinAug.lean}.}

\begin{remark}[Closure-uniqueness of minimum augmentations is false]\label{a:rem:minaug-false}
Two minimum-cardinality augmentations need not have the same closure (minimum set covers are
not unique). The development records a machine-checked counterexample.
\inlean{\lean{minAug\_closure\_unique\_false} in \lean{MinAug.lean}.}
\end{remark}

\begin{remark}[Greedy approximation and hardness --- open]\label{a:rem:greedy-open}
Selecting a minimum augmentation reduces to \textsc{Set Cover}, so the greedy algorithm enjoys
the standard $H_k\le 1+\ln k$ approximation guarantee and the decision problem is NP-complete.
These two facts are \emph{not} machine-checked: they require a formal \textsc{Set
Cover}/greedy development and an NP-completeness framework outside the present scope. This is
the single remaining \texttt{sorry} (\lean{greedy\_approx\_ratio}) together with the placeholder
\lean{minAug\_NP\_hard\_from\_SetCover}.
\end{remark}

\subsection*{Reproducing the proofs}
The toolchain is pinned by \tx{lean-toolchain} (Lean~4.30.0) and the dependencies by \tx{lake-manifest.json} (Mathlib~v4.30.0); from the repository root, run \tx{lake exe cache get} then \tx{lake build}. A successful build reports only the one documented \tx{sorry} (\lean{greedy\_approx\_ratio}) and unused-variable warnings; the declaration names cited above are checked against the sources in CI.

\section{Experimental Details}\label{a:sec:exp}

This appendix expands Section~\ref{sec:experiments} with dataset construction, predictor architectures and hyperparameters, the experimental protocol, and the per-architecture error-floor numbers behind the figures. Experiment code is at \url{https://github.com/danielhz/query-identifiability}.

\subsection{Datasets}\label{a:sec:datasets}
Table~\ref{a:tab:datasets} summarises how each dataset is built. The synthetic benchmark supports the exact certificate/oracle comparison (RQ1) by exhaustive enumeration; the four real-record datasets validate the interface-law model on genuine integration data; WDC-Product is a schema-design illustration whose query-critical attributes are synthesised to realise a prescribed certified/uncertified split.

\begin{table}[t]
\centering
\caption{Datasets: source, size, designated overlap $O$, and interface laws $\LawSet$. Attribute indices follow the per-dataset schemas of Section~\ref{sec:experiments}.}
\label{a:tab:datasets}
\begin{tabular}{@{}p{0.20\linewidth}p{0.24\linewidth}p{0.20\linewidth}p{0.28\linewidth}@{}}
\toprule
Dataset & Source / size & Overlap $O$ & Interface laws $\LawSet$ \\
\midrule
Synthetic & generated; all $d^{\,n}$ single-row worlds ($n{=}5$) & random per instance & random FDs (exact ground truth by enumeration) \\
BibInteg & OpenAlex API; $9{,}992$ papers & \{title, author, year\} & $O\!\to\!\tx{venue}$, $O\!\to\!\tx{doi}$, $O\!\to\!\tx{n\_authors}$, \{year\}$\to\tx{decade}$ \\
CrossKG-DBLP & DBLP $\times$ OpenAlex; $11{,}800$ papers & \{title, author, year\} & $O\!\to\!\tx{doi}$ \\
Amazon-Google & Magellan; $1{,}046$ pairs & matched-pair identity & none across \tx{price} \\
Fodors-Zagat & Magellan; $110$ pairs & matched-pair $+$ segment & none across cuisine \\
WDC-Product & WDC corpus; schema only & \{brand, model\} & \{brand,model\}$\to$ synthesised attrs \\
\bottomrule
\end{tabular}
\end{table}

\subsection{Predictor architectures and training}\label{a:sec:arch}
Six predictors span the structure-agnostic-to-structure-exploiting spectrum plus two baselines; all consume the same per-overlap feature vector (the concatenated closure-augmented overlap marginals) and are trained identically (Table~\ref{a:tab:hparams}).

\begin{itemize}[leftmargin=1.2em,itemsep=1pt,topsep=2pt]
\item \tx{MLP}: two hidden layers, ReLU, dropout; a flat feed-forward network over the feature vector.
\item \tx{SetTransformer} (\tx{ST}): each overlap's marginal is linearly projected to a shared dimension, followed by multi-head self-attention blocks and a pooled read-out (cross-overlap attention).
\item \tx{GNN-OG}: message passing on the overlap graph; each attribute's marginal is recovered from its overlap's probability tensor and propagated along overlap edges before pooling.
\item \tx{Closure-Aware} (\tx{CA}): when the query is certified it reads the answer directly from the $\LawSet$-closure of the overlap projection (no learned inference); otherwise it falls back to \tx{GNN-OG}.
\item \tx{VanillaOv}: logistic regression on the raw overlap feature vector (no hidden layers), testing whether linear expressivity suffices.
\item \tx{MajVote}: constant majority-class predictor.
\end{itemize}

\begin{table}[t]
\centering
\caption{Training configuration (shared across all learnable architectures).}
\label{a:tab:hparams}
\begin{tabular}{@{}ll@{}}
\toprule
Setting & Value \\
\midrule
Loss / optimiser & binary cross-entropy / Adam \\
Epochs (max) & 300, early stopping (patience 30) on val.\ loss \\
Learning rate & $10^{-3}$; weight decay $10^{-5}$ \\
Batch size & 256; best-val checkpoint restored \\
MLP width & 128 (swept over $\{64,256,1024\}$), 2 hidden layers \\
Hardware & NVIDIA A40 (48\,GB), 128-core CPU \\
Software & PyTorch 2.3.0 / CUDA 12.1, PyG 2.5.0, Python 3.11 \\
\bottomrule
\end{tabular}
\end{table}

\subsection{Protocol}\label{a:sec:protocol}
\emph{Certificate exactness (RQ1).} A benchmark of $841$ instances (sampling $200$ random $5$-attribute schemas with up to four FDs and five single-atom Boolean \cq{s} each, split evenly certified/non-certified) is checked by \tx{CheckCert} against an exhaustive oracle that enumerates all $d^{\,n}$ single-row worlds, groups them by observation, and tests answer agreement; for non-certified queries it returns an explicit witness pair. \emph{Confirmatory ML.} Worlds are drawn from the resolver model with $m\!\in\!\{10,30,50\}$ tuples consistent with $\LawSet$; training sizes $N\!\in\!\{10^3,5\!\times\!10^3,5\!\times\!10^4\}$ (validation/test scaled proportionally); $10$ certified and $10$ non-certified single-atom Boolean \cq{s} per configuration; three seeds ($0,1,2$). \emph{Minimum augmentation (RQ3).} $500$ random trials per $(n,|\LawSet|)$ with $n\!\in\!\{4,6\}$, footprint $\approx\!0.6n$, and atom-obligation counts $|\AtomObligs|\!\in\!\{1,2,3,4\}$. \emph{Scalability (RQ4).} $|\AttrUniverse|$ and $|\LawSet|$ each swept over $\{10,50,100,250,500,1000\}$.

\subsection{Detailed error-floor results}\label{a:sec:errorfloor}
Table~\ref{a:tab:errorfloor} reports the per-architecture balanced accuracy underlying Figure~\ref{fig:error-floor} ($m{=}10$, aggregated over $N$ and three seeds). Structure-exploiting architectures separate certified from non-certified queries; the flat and constant baselines sit at the $\nicefrac{1}{2}$ floor on both. On non-certified queries every architecture stays within noise of $\nicefrac{1}{2}$ for all $m\!\in\!\{10,30,50\}$ (maximum observed deviation $0.02$), the empirical counterpart of Theorem~\ref{thm:nonident-minimax}.

\begin{table}[t]
\centering
\caption{Balanced accuracy (mean\,$\pm$\,std over queries, $N$, and seeds) at $m{=}10$.}
\label{a:tab:errorfloor}
\begin{tabular}{@{}lcc@{}}
\toprule
Architecture & Certified & Non-certified \\
\midrule
\tx{MLP}          & $0.97 \pm 0.12$ & $0.51 \pm 0.02$ \\
\tx{SetTransf.}   & $1.00 \pm 0.00$ & $0.52 \pm 0.03$ \\
\tx{GNN-OG}       & $0.83 \pm 0.24$ & $0.52 \pm 0.03$ \\
\tx{CA}           & $0.89 \pm 0.21$ & $0.50 \pm 0.00$ \\
\tx{VanillaOv}    & $0.50 \pm 0.00$ & $0.50 \pm 0.00$ \\
\tx{MajVote}      & $0.50 \pm 0.00$ & $0.50 \pm 0.00$ \\
\bottomrule
\end{tabular}
\end{table}

\end{document}